\documentclass[journal]{IEEEtran}
\usepackage{amsmath}
\usepackage{amsthm}
\usepackage{amsfonts}
\usepackage{amssymb}
\usepackage{multicols}
\usepackage{multirow}
\usepackage{authblk}
\usepackage{graphicx}
\usepackage{epstopdf}
\usepackage{mathtools}
\usepackage{tikz}
\usetikzlibrary{shapes,arrows}
\tikzstyle{line}=[draw]
\usepackage{bigstrut}

\theoremstyle{remark}
\newtheorem{thm}{Theorem}

\newtheorem{lem}{Lemma}

\newtheorem{cor}{Corollary}
\newtheorem{defn}{Definition}
\newtheorem{exmp}{Example}

\newtheorem{rem}{Remark}

\title{Optimal Linear Broadcast Rates of the Two-Sender Unicast Index Coding Problem with Fully-Participated Interactions}
\author{Chinmayananda Arunachala, Vaneet Aggarwal, and B. Sundar Rajan. \thanks{C. Arunachala and B. S. Rajan are with the Dept.
		of Electrical Communication Engg., Indian Institute of Science, Bengaluru
		560012, KA, India, email: \{chinmayanand,bsrajan\}@iisc.ac.in. 	V. Aggarwal is with the School of Industrial
		Engineering at Purdue University, West Lafayette, IN, USA 47907. He is also with the
		Dept. of Electrical Communication Engg., Indian Institute of Science, Bengaluru
		560012, KA, India, email:
		vaneet@purdue.edu.}}

\begin{document}
\maketitle
\begin{abstract}
	The two-sender unicast index coding problem consists of finding optimal coded transmissions from the two senders which collectively know the messages demanded by all the receivers. Each receiver demands a unique message. One important  class of this problem consists of the message sets at the senders and the side-information at the receivers satisfying  \emph{fully-participated interactions}. This paper provides optimal linear broadcast rates and corresponding code constructions for all the possible cases of the two-sender unicast index coding problem with fully-participated interactions. The optimal linear broadcast rate and the corresponding code for the two-sender problem are given in terms of those of the three single-sender unicast problems associated with the two-sender problem. Optimal linear broadcast rates of   two-sender problems with fully-participated interactions provide  lower bounds for the optimal linear broadcast rates of many related two-sender problems with \emph{partially-participated interactions}. Proof techniques used to obtain the results for the two-sender problem are shown to be useful in obtaining the results for some cases of the  multi-sender unicast index coding problem.
\end{abstract}

\section{Introduction}
The classical index coding problem (ICP) introduced in \cite{BK} consists of a sender who has to broadcast coded messages to a set of receivers, where each receiver has some subset of messages demanded by other receivers (also known as the side-information of that receiver). The sender knows the side-information of each receiver and uses this knowledge to encode the messages demanded by them. This encoding  reduces the number of broadcast transmissions compared to the naive transmission of each message. The receivers make use of their side-information and the broadcast transmissions to decode their demanded messages. In many practical scenarios, messages are distributed among multiple senders to reduce the delay in content delivery. For example, content is delivered using large storage capacity nodes called caching helpers in cellular networks \cite{KAC}. Data is also distributed and stored over multiple storage nodes in distributed storage networks \cite{luo2016coded,xiang2016joint}. In some scenarios, each sender can have access to only a subset of messages due to data storage limits or errors in the reception of some messages over noisy channels.  Hence, the multi-sender ICP is of practical significance. 

Ong et al. \cite{SUOH} studied a class of multi-sender ICPs, where each receiver knows a unique message and demands a subset of other messages. They provide an iterative algorithm which gives different lower bounds for the optimal codelength based on the \emph{strongly-connected component} of the \emph{information-flow graph} selected in each iteration. There is little insight on the tightness of lower bounds for the optimal codelength and how good the algorithm works. Thapa et al. \cite{COJ} extended some single-sender index coding schemes based on graph theory to the two-sender unicast ICP (TUICP), where each receiver demands a unique message. No results on the optimality of the coding schemes and the tightness of the gap between the optimal codelength and the codelengths obtained by the proposed schemes are provided, for any general class of the TUICP. Several works provide inner and outer bounds for the capacity region of variations of multi-sender ICP \cite{sadeghi2016distributed,YPFK,MOJ2,MOJ}. These works assume that there are links with fixed finite capacities from every sender to every receiver in contrast to the previous works. They use variations of random coding to provide the  bounds for the capacity region. Schemes which improve the tightness of the bounds are also provided. 

Thapa et al. \cite{CTLO} studied the TUICP using the two-sender graph coloring of the \emph{confusion graph} to obtain the optimal broadcast rate with $t$-bit messages for any finite $t$. The TUICP was jointly described by the \emph{side-information digraph} and the messages present at each sender. It was  analyzed using three single-sender sub-problems  described by the three vertex-induced sub-digraphs of the side-information digraph respectively. The partition of the side-information digraph into the three vertex-induced sub-digraphs depends on the availability of messages at the two senders. The TUICP was classified into 64 types based on the \emph{interactions} among these sub-digraphs. The type of interactions among these sub-digraphs was described by the \emph{interaction digraph} of the associated side-information digraph based on the availability of messages at the senders. The 64 possible interaction digraphs were broadly classified into two cases: Case I and Case II. Case I consists of all the possible acyclic digraphs on three vertices. Case II was further classified into five subcases. For some cases, the optimal broadcast rates with $t$-bit messages for any finite $t$ and the corresponding code constructions were provided in terms of those of the three related single-sender sub-problems. Upper bounds were provided for other cases. Similarly, the optimal broadcast rates (as $t \rightarrow \infty$) were provided for some cases. Upper bounds were given for other cases. Thus, the complexity of finding the optimal results for the TUICP was reduced to that of finding the same for the single-sender unicast ICP. 

This paper provides the optimal linear broadcast rate with $t$-bit messages for any finite $t$ and the corresponding code construction, and the optimal linear broadcast rate,  for all the cases of the TUICP with fully-participated interactions. These results for the TUICP are given in terms of those of the three constituent single-sender unicast ICPs. In general, linear encoding schemes are of interest as it is easy to encode the messages. The optimal linear broadcast rate with $t$-bit messages for any finite $t$ and the optimal linear broadcast rate of any two-sender problem with fully-participated interactions provide lower bounds for the corresponding  results of many two-sender problems with partially-participated interactions, having the same associated single-sender sub-problems as the original two-sender problem. These results help in establishing the corresponding optimal results of many two-sender problems with partially-participated interactions. 

As shown in Section IV, it is difficult to classify any multi-sender unicast index coding problem with fully-participated interactions based on its interaction digraph. However, optimal results and the proof techniques used to obtain the results of the TUICP with fully-participated interactions can be used to obtain the optimal results for some classes of the multi-sender problem with fully-participated interactions. Further, sub-optimal results can be obtained by partitioning the multi-sender problem into  multi-sender sub-problems with less number of senders for which optimal results can be easily found. Hence, it is important to know the optimal results of the TUICP with fully-participated interactions.   

The key results of this paper are summarized as follows. 
\begin{itemize}
	\item Optimal linear broadcast rate with $t$-bit messages for any finite $t$, corresponding code construction, and optimal linear broadcast rate are provided for all the cases of the TUICP with fully-participated interactions. 
	\item For Cases I and II-A, the same results are also shown to be valid for any two-sender problem with any partially-participated interactions.
	\item The proof techniques and the results obtained for the TUICP are used to derive the optimal results for some cases of the multi-sender ICP. 
\end{itemize}
\par The remainder of the paper is organized as follows. Section II introduces the proposed  problem and provides the required definitions and notations. Section III provides the main results of the paper. Section IV illustrates the application of the proof techniques used in Section III to solve some cases of the multi-sender problem. Section V concludes the paper with directions for future work.

\section{Problem Formulation and Definitions}

In this section, we formulate the two-sender unicast index coding problem, and provide the required notations and definitions used in this paper. 
\par  The set $\{1,2,\cdots, n\}$ is denoted as $[n]$. An instance of the two-sender unicast index coding problem (TUICP), consists of $m$ independent messages given by the set $\mathcal{M} =\{{\bf{x}}_1,{\bf{x}}_2,\cdots,{\bf{x}}_{m}\}$, where ${\bf{x}}_i \in \mathbb{F}_2^{t \times 1}$, $\forall i \in [m]$, and a positive integer $t \geq 1$. There are $m$ receivers. The $i$th receiver demands ${\bf{x}}_i$ and has $\mathcal{K}_i \subseteq \mathcal{M} \setminus \{{\bf{x}}_i\}$ as its side-information. The $s$th sender is denoted by $S_{s}$, $s \in \{1,2\}$. $S_{s}$ possesses the message set  $\mathcal{M}_{s}$ such that $\mathcal{M}_{s} \subset \mathcal{M}$ and $\mathcal{M}_{1} \cup \mathcal{M}_{2}=\mathcal{M}$. Each sender knows the identity of the messages present with the other sender and transmits through the same noiseless broadcast channel. Transmissions from different senders are orthogonal in time. The  single-sender unicast ICP is a special case of the TUICP, where $\mathcal{M}_{1}=\mathcal{M}$ and  $\mathcal{M}_{2}=\Phi$. 

Given an instance of the TUICP, each codeword of a two-sender index code consists of two sub-codewords broadcasted by the two senders respectively. An encoding function for the sender $S_{s}$ is given by $\mathbb{E}_{s}:\mathbb{F}_{2}^{|\mathcal{M}_{s}|t \times 1} \rightarrow   \mathbb{F}_{2}^{p_{s} \times 1}$, such that $\mathcal{C}_s=\mathbb{E}_{s}(\mathcal{M}_s)$, where $p_s$ is the length (number of bits) of the sub-codeword $\mathcal{C}_s$ transmitted by $S_s$, $s \in \{1,2\}$. The sub-codewords from the two senders are transmitted one after the other. The $i$th receiver has a decoding function given by $\mathbb{D}_{i}:\mathbb{F}_{2}^{(p_{1}+p_{2}+|\mathcal{K}_{i}|t) \times 1} \rightarrow   \mathbb{F}_{2}^{t \times 1}$, such that ${\bf{x}}_i = \mathbb{D}_{i}(\mathcal{C}_1,\mathcal{C}_2,\mathcal{K}_i)$, $i \in [m]$, i.e., it can decode ${\bf{x}}_i$ using its side-information and the received codeword consisting of $\mathcal{C}_1$ and $\mathcal{C}_2$. For the single-sender unicast ICP, $\mathcal{M}_{2}=\Phi$.  Hence, $p_2=0$. In this case, we assume that only $\mathbb{E}_1$ exists.
\begin{figure*}[!htbp]
	\begin{center}
		\includegraphics[width=41pc]{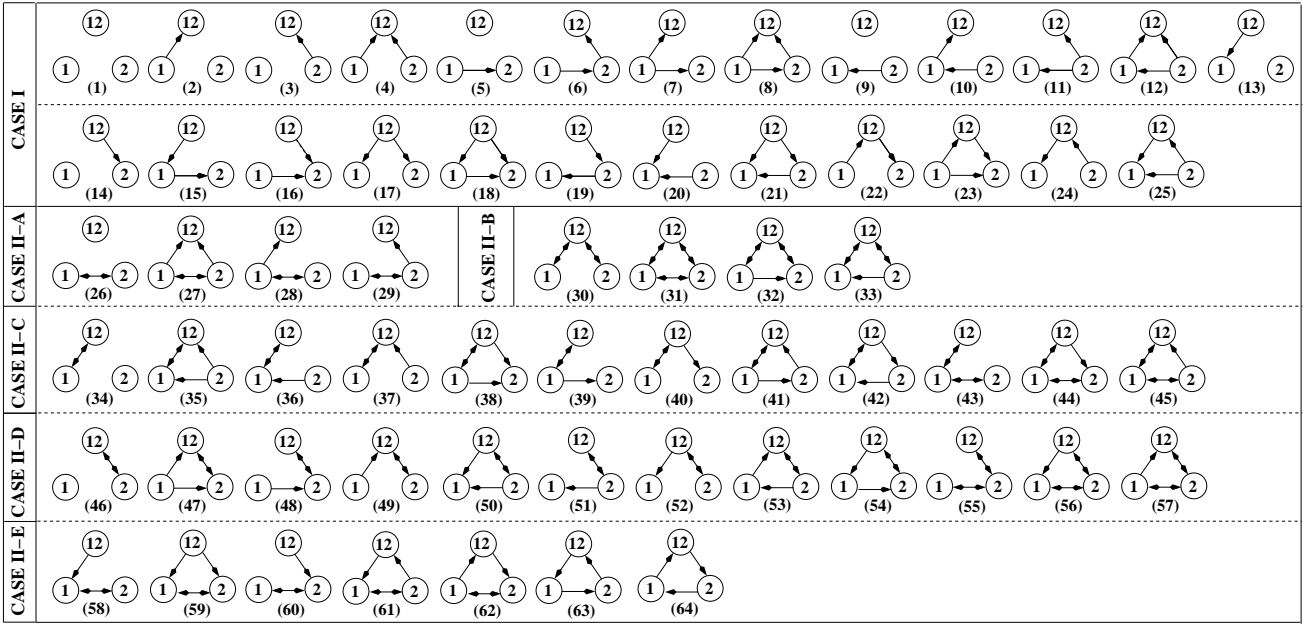}
		\caption{Enumeration of all the possible interactions between the sub-digraphs $\mathcal{D}_1$, $\mathcal{D}_2$, and $\mathcal{D}_{\{1,2\}}$, denoted by the interaction digraph $\mathcal{H}$.}
		\label{interenum}
	\end{center}
	\hrule
\end{figure*}
\par We define the linear broadcast rate of an index code, the optimal linear broadcast rate with $t$-bit messages for any finite $t$, and the optimal linear broadcast rate of the two-sender problem, which consider only linear encoding schemes. 

\begin{defn}[Linear broadcast rate]
	An index code for an instance of the TUICP is said to be linear, if the encoding functions are linear transformations. Let ${\bf{x}}^s \in \mathbb{F}_2^{t|\mathcal{M}_s| \times 1}$ be the concatenated message vector obtained by concatenating the $t$-bit  messages available at the sender $S_s,$ $s \in \{1,2\}$. Let $S_s$ broadcast $\mathcal{C}_s={\bf{G}}^s{\bf{x}}^s$ of length $p_s$, where ${\bf{G}}^s \in \mathbb{F}_2^{p_s \times t|\mathcal{M}_s|}$. For a single-sender problem $p_2=0$ and $S_1$ alone  broadcasts the index code. The linear broadcast rate of the index code  described by $\{{\bf{G}}^s\}$ is given by $p_{t}^{l} \triangleq \frac{p_1+p_2}{t}$. 
\end{defn}
\begin{defn}[Optimal linear broadcast rate with $t$-bit messages for any finite $t$]
	The optimal linear broadcast rate of a unicast  ICP (single-sender or two-sender) with $t$-bit  messages for any finite $t$ is given by $\beta_{t}^{l} \triangleq \underset{\{{\bf{G}}^s\}}{min}$  $p_{t}^{l}$.
\end{defn}
\begin{defn}[Optimal linear broadcast rate]
	The optimal linear broadcast rate (over all $t$) of a unicast ICP (single-sender or two-sender) is defined as $\beta^l \triangleq \underset{t}{inf} \beta_{t}^{l} = \underset{t \rightarrow \infty}{lim} \beta_{t}^{l}$.
	
	The limit exists and is equal to the infimum due to the subadditivity of $t\beta_t^l$ and Fekete's lemma \cite{fekete}. 
\end{defn}

We state some definitions from graph theory \cite{DBW}, that will be used in this paper. 

A directed graph (also called digraph) given by $\mathcal{D}=(\mathcal{V}(\mathcal{D}),\mathcal{E}(\mathcal{D}))$, consists of a set of vertices $\mathcal{V}(\mathcal{D})$, and a set of edges $\mathcal{E}(\mathcal{D})$ which is a set of ordered pairs of vertices. A sub-digraph $\mathcal{G}$ of a digraph $\mathcal{D}$ is a digraph, whose vertex set satisfies $\mathcal{V}(\mathcal{G}) \subseteq \mathcal{V}(\mathcal{D})$, and the edge set satisfies $\mathcal{E}(\mathcal{G}) \subseteq \mathcal{E}(\mathcal{D})$. The sub-digraph $\mathcal{G}$ of $\mathcal{D}$ induced by the vertex set $\mathcal{V}(\mathcal{G})$ is the digraph whose vertex set is $\mathcal{V}(\mathcal{G})$,  and the edge set is given by  $\mathcal{E}(\mathcal{G})=\{(u, v): u,v \in \mathcal{V}(\mathcal{G}), (u, v) \in \mathcal{E}(\mathcal{D})\}$. 

A directed path in a digraph $\mathcal{D}$ is a sequence of distinct vertices $\{v_1,\cdots,v_r\}$, such that $(v_i,v_{i+1}) \in \mathcal{E}(\mathcal{D})$, $\forall i \in [r-1]$. A cycle in a digraph $\mathcal{D}$ is a sequence of distinct vertices $(v_{1},\cdots,v_{c})$, such that $(v_{i},v_{i+1}) \in \mathcal{E}(\mathcal{D})$,  $\forall i \in [c-1]$, and $(v_{c},v_{1}) \in \mathcal{E}(\mathcal{D})$. A digraph with no cycles is called acyclic.

\begin{defn}[Topological ordering, \cite{DEK}]
	A topological ordering of a digraph $\mathcal{D}$ is a labelling of its vertices using the numbers in $\{1,2,\cdots,|\mathcal{V}(\mathcal{D})|\}$, such that for every edge $(u,v) \in \mathcal{E}(\mathcal{D})$,  $u < v$, where $u,v \in \{1,2,\cdots,|\mathcal{V}(\mathcal{D})|\}$.
\end{defn}

For any unicast ICP (single-sender or multi-sender), the knowledge of  side-information and demands of all the receivers is represented by the side-information digraph given by $\mathcal{D}=(\mathcal{V}(\mathcal{D}),\mathcal{E}(\mathcal{D}))$, where the vertex set is given by $\mathcal{V}(\mathcal{D})=\{v_{1},\cdots,v_{m}\}$. The vertex $v_{i}$ represents the $i$th receiver which demands the message ${\bf{x}}_i$. Due to the one-to-one relationship between the $i$th receiver and ${\bf{x}}_i$, $v_{i}$ also represents ${\bf{x}}_i$. Hence, we refer to  $v_i$ as the $i$th message, the $i$th receiver, and the  $i$th vertex  interchangeably. The edge set is given by $\mathcal{E}(\mathcal{D})=\{(v_{i},v_{j}): {\bf{x}}_{j} \in \mathcal{K}_{i}, i,j \in [m]\}$. 
Consider the message sets $\mathcal{P}_{1} = \mathcal{M}_{1} \setminus \mathcal{M}_{2}$ and $\mathcal{P}_{2} = \mathcal{M}_{2} \setminus \mathcal{M}_{1}$, which are available only with $S_{1}$ and $S_{2}$ respectively. The messages available with both the senders are given by  $\mathcal{P}_{\{1,2\}} = \mathcal{M}_{1} \cap \mathcal{M}_{2}$. Let $m_S=|\mathcal{P}_{S}|$, for any non-empty set $S \subseteq  \{1,2\}$. We represent any  singleton set without $\{\}$. For example, $\{1\}$ is written as $1$. Let $\mathcal{P}=(\mathcal{P}_{1},\mathcal{P}_{2},\mathcal{P}_{\{1,2\}})$. Any TUICP $\mathcal{I}$ can be described in terms of the two tuple $(\mathcal{D},\mathcal{P})$, as $\mathcal{I}(\mathcal{D},\mathcal{P})$. The broadcast rates $\beta^l$ and $\beta_{t}^l$ of $\mathcal{I}(\mathcal{D},\mathcal{P})$ are denoted by $\beta^l(\mathcal{D},\mathcal{P})$ and $\beta_{t}^l(\mathcal{D},\mathcal{P})$ respectively. For an instance of the single-sender unicast ICP with the side-information digraph $\mathcal{D}$, the broadcast rates $\beta^l$ and $\beta_{t}^l$ are denoted by $\beta^l(\mathcal{D})$ and $\beta_{t}^l(\mathcal{D})$ respectively.

The TUICP was analyzed using three sub-digraphs (equivalently sub-problems) induced by three disjoint vertex sets of the side-information digraph respectively \cite{CTLO}. Let $\mathcal{D}_S$ be the sub-digraph of $\mathcal{D}$, induced by the vertices $\{v_j: {\bf{x}}_j \in \mathcal{P}_{S}, j \in [m]\}$, for any non-empty set $S \subseteq \{1,2\}$. If there exists an edge from some vertex in $\mathcal{V}(\mathcal{D}_{S})$ to some vertex in $\mathcal{V}(\mathcal{D}_{S'})$, in the side-information digraph $\mathcal{D}$, for non-empty sets $S,S' \subseteq \{1,2\}, S \neq S'$, then we say that there is an interaction from $\mathcal{D}_{S}$ to $\mathcal{D}_{S'}$, and denote it as $\mathcal{D}_{S} \rightarrow \mathcal{D}_{S'}$. We say that the interaction $\mathcal{D}_{S} \rightarrow \mathcal{D}_{S'}$ is fully-participated, if there are edges from every vertex in $\mathcal{V}(\mathcal{D}_{S})$ to every vertex in $\mathcal{V}(\mathcal{D}_{S'})$. Otherwise, it is said to be a partially-participated interaction. We say that the TUICP has fully-participated interactions if all the existing interactions are fully-participated interactions. Consider the digraph $\mathcal{H}$ with  $\mathcal{V}(\mathcal{H}) = \{1,2,\{1,2\}\}$ and $\mathcal{E}(\mathcal{H})=\{(S,S') | \mathcal{D}_{S} \rightarrow \mathcal{D}_{S'}, S,S' \in \mathcal{V}(\mathcal{H})\}$. Let  $f: \mathcal{V}(\mathcal{D}) \rightarrow \mathcal{V}(\mathcal{H})$ be a function such that  $f(v)=S$ if $v \in \mathcal{V}(\mathcal{D}_{S})$. We call the digraph $\mathcal{H}$ as the interaction digraph of the side-information digraph $\mathcal{D}$. The edges $(S,S')$ and $(S',S)$ in any interaction digraph are denoted by a single edge with arrows at both the ends. There are 64 possibilities for the digraph $\mathcal{H}$ as shown in Figure \ref{interenum}, which were enlisted  and classified in \cite{CTLO}. The vertex representing the set $\{1,2\}$ is written as $12$ for brevity. The number written below each interaction digraph in the figure is used as the subscript to denote the specific interaction digraph. Note that different TUICPs with the same message tuple $\mathcal{P}$ can have the same interaction digraph. Any side-information digraph $\mathcal{D}$ with interaction digraph $\mathcal{H}_{k}$ is denoted by $\mathcal{D}^{k}$, $k \in \{1,2,\cdots,64\}$. For any TUICP $\mathcal{I}(\mathcal{D}^k,\mathcal{P})$, the corresponding sub-digraph $\mathcal{D}_S$ for a non-empty set $S \subseteq \{1,2\}$ is denoted by $\mathcal{D}_{S}^{k,\mathcal{P}}$. Any TUICP $\mathcal{I}(\mathcal{D}^k,\mathcal{P})$ is analyzed using the three single-sender unicast ICPs with the side-information digraphs $\mathcal{D}_{S}^{k,\mathcal{P}}$, for non-empty sets $S \subseteq \{1,2\}$. Note that all the possible interaction digraphs are classified into two cases broadly: Case I and Case II. Case I consists of acyclic interaction digraphs. Case II was further classified into five subcases as shown in Figure \ref{interenum}. We illustrate the above definitions using an example.    
\begin{figure}[!htbp]
	\begin{center}
		\begin{tikzpicture}
		[place/.style={circle,draw=black!100,thick}]
		\node at (-5.5,.2) [place] (1c) {1};
		\node at (-6.7,-.7) [place] (5c) {5};
		\node at (-6,-1.8) [place] (4c) {4};
		\node at (-4.7,-1.8) [place] (3c) {3};
		\node at (-4.2,-.7) [place] (2c) {2};
		\node at (-3.6,-1.8) [place] (1h) {1}; 	 	
		\node at (-2.1,-1.8) [place] (2h) {2};
		\node at (-2.8,-.7) [place] (3h) {12};	
		\node at (-0.5,-.2) [place] (1'h) {1}; 	 	
		\node at (1,-.2) [place] (2'h) {2};
		\node at (1,-1.1) [place] (3'h) {3};    		     	
		\node at (-0.5,-1.1) [place] (4'h) {4};
		\node at (0.2,-1.9) [place] (5'h) {5}; 		     	
		\draw (-5.4,-2.5) node {$\mathcal{D}$};
		\draw (-2.9,-2.5) node {$\mathcal{H}$};
		\draw (-1.3,-.2) node {$\mathcal{D}_1$};
		\draw (-1.3,-1.1) node {$\mathcal{D}_2$};     		  
		\draw (-0.9,-2) node {$\mathcal{D}_{\{1,2\}}$};
		\draw [thick,->] (1h) to [out=0,in=-180] (2h);
		\draw [thick,->] (2h) to [out=110,in=-45] (3h);     		     	
		\draw [thick,<->] (3h) to [out=-140,in=65] (1h);    		     	
		\draw [thick,->] (1'h) to [out=0,in=-180] (2'h);
		\draw [thick,->] (3'h) to [out=160,in=20] (4'h);     		     	
		\draw [thick,->] (4'h) to [out=-20,in=-160] (3'h);     		     	
		\draw [thick,->] (1c) to [out=-5,in=120] (2c);
		\draw [thick,->] (2c) to [out=-100,in=55] (3c);     		     	
		\draw [thick,->] (3c) to [out=160,in=20] (4c);    		     	
		\draw [thick,->] (4c) to [out=-20,in=-160] (3c);
		\draw [thick,->] (3c) to [out=135,in=-35] (5c);     		     	
		\draw [thick,->] (1c) to [out=-50,in=105] (3c);    	
		\draw [thick,->] (1c) to [out=-170,in=65] (5c);  
		\draw [thick,->] (5c) to [out=30,in=-130] (1c);  
		\draw [thick,->] (5c) to [out=0,in=180] (2c);	     	
		\end{tikzpicture}
		\caption{Example to illustrate the interaction digraph $\mathcal{H}$ and the sub-digraphs $\mathcal{D}_1$, $\mathcal{D}_2$, and $\mathcal{D}_{\{1,2\}}$ of a given side-information digraph $\mathcal{D}$.}
		\label{examp1}
	\end{center}
\end{figure}
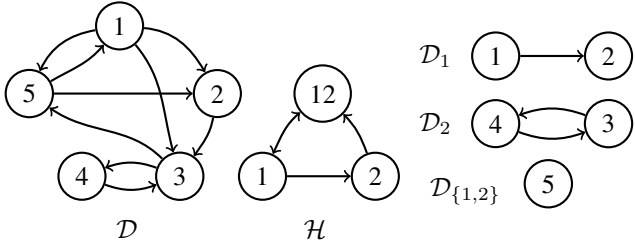
\begin{exmp}
	Consider the TUICP with $m=5$ messages, where the $i$th receiver demands ${\bf{x}}_{i}$. Sender $S_1$ has $\mathcal{M}_1=\{{\bf{x}}_1,{\bf{x}}_2,{\bf{x}}_5\}$. $S_2$ has $\mathcal{M}_2=\{{\bf{x}}_3,{\bf{x}}_4,{\bf{x}}_5\}$. Hence, $\mathcal{P}_1=\{{\bf{x}}_1,{\bf{x}}_2\}$, $\mathcal{P}_2=\{{\bf{x}}_3,{\bf{x}}_4\}$, and $\mathcal{P}_{\{1,2\}}=\{{\bf{x}}_5\}$. The side-information of each receiver is given as follows: $\mathcal{K}_{1}=\{{\bf{x}}_{2},{\bf{x}}_{3},{\bf{x}}_{5}\}$, $\mathcal{K}_{2}=\{{\bf{x}}_{3}\}$, $\mathcal{K}_{3}=\{{\bf{x}}_{4},{\bf{x}}_{5}\}$, $\mathcal{K}_{4}=\{{\bf{x}}_{3}\}$, $\mathcal{K}_{5}=\{{\bf{x}}_{1},{\bf{x}}_{2}\}$. The side-information digraph  $\mathcal{D}$ and the corresponding interaction digraph  $\mathcal{H}$ are shown in Figure \ref{examp1}.  The vertex-induced sub-digraphs $\mathcal{D}_1$, $\mathcal{D}_2$, and $\mathcal{D}_{\{1,2\}}$ induced by the messages in $\mathcal{P}_1$, $\mathcal{P}_2$, and $\mathcal{P}_{\{1,2\}}$ respectively are also shown in the figure. Note that the interaction $\mathcal{D}_{\{1,2\}} \rightarrow \mathcal{D}_{1}$ is fully-participated. Others are partially-participated interactions. The interaction digraph shown in Figure \ref{examp1} is $\mathcal{H}_{41}$ as given in Figure \ref{interenum}. Hence, the side-information digraph $\mathcal{D}$ can also be denoted as $\mathcal{D}^{41}$.
\end{exmp}

\par The following notations are required for the construction of a two-sender index code from single-sender index codes. Let $\mathcal{C}_1$ and $\mathcal{C}_2$ be  two codewords of length $l_1$ and $l_2$ respectively. $\mathcal{C}_1 \oplus \mathcal{C}_2$ denotes the bit-wise XOR of $\mathcal{C}_1$ and $\mathcal{C}_2$ after zero-padding the shorter message at the least significant positions to match the length of the longer message. The resulting length of the codeword is  $max(l_1,l_2)$. For example, if $\mathcal{C}_1=1010$, and $\mathcal{C}_2=110$, then $\mathcal{C}_1 \oplus \mathcal{C}_2 = 0110$.  $\mathcal{C}[a:b]$ denotes the vector obtained by picking the bits from bit position $a$ to bit position $b$, starting from the most significant position of the codeword $\mathcal{C}$, with $a,b \in [l]$, $l$ being the length of $\mathcal{C}$. For example $\mathcal{C}_1[2:4]=010$.

\section{Main Results}

\begin{table*}[t]
	\centering
	\begin{tabular}{|c|c|c|c|c|}
		\hline
		CASE & 
		$\beta^l(\mathcal{D}^k,\mathcal{P})$ & $\beta_{t}^l(\mathcal{D}^k,\mathcal{P})$ \\  \hline
		I & $\beta^{l}(\mathcal{D}_{1}^{k,\mathcal{P}})+\beta^{l}(\mathcal{D}_{2}^{k,\mathcal{P}})+\beta^{l}(\mathcal{D}_{\{1,2\}}^{k,\mathcal{P}})!$ & $\beta_{t}^{l}(\mathcal{D}_{1}^{k,\mathcal{P}})+\beta_{t}^{l}(\mathcal{D}_{2}^{k,\mathcal{P}})+\beta_{t}^{l}(\mathcal{D}_{\{1,2\}}^{k,\mathcal{P}})!$ \\ \hline
		II-A & $\beta^{l}(\mathcal{D}_{1}^{k,\mathcal{P}})+\beta^{l}(\mathcal{D}_{2}^{k,\mathcal{P}})+\beta^{l}(\mathcal{D}_{\{1,2\}}^{k,\mathcal{P}})!$ & $\beta_{t}^{l}(\mathcal{D}_{1}^{k,\mathcal{P}})+\beta_{t}^{l}(\mathcal{D}_{2}^{k,\mathcal{P}})+\beta_{t}^{l}(\mathcal{D}_{\{1,2\}}^{k,\mathcal{P}})!$ \\ \hline
		II-B & 
		$max\{\beta^{l}(\mathcal{D}_{\{1,2\}}^{k,\mathcal{P}}),\beta^{l}(\mathcal{D}_{1}^{k,\mathcal{P}})+\beta^{l}(\mathcal{D}_{2}^{k,\mathcal{P}})\}$ & $max\{\beta_{t}^{l}(\mathcal{D}_{\{1,2\}}^{k,\mathcal{P}}),\beta_{t}^{l}(\mathcal{D}_{1}^{k,\mathcal{P}})+\beta_{t}^{l}(\mathcal{D}_{2}^{k,\mathcal{P}})\}$ \\ \hline
		II-C & 
		$\beta^{l}(\mathcal{D}_{2}^{k,\mathcal{P}})+max\{\beta^{l}(\mathcal{D}_{1}^{k,\mathcal{P}}),\beta^{l}(\mathcal{D}_{\{1,2\}}^{k,\mathcal{P}})\}$ & $\beta_{t}^{l}(\mathcal{D}_{2}^{k,\mathcal{P}})+max\{\beta_{t}^{l}(\mathcal{D}_{1}^{k,\mathcal{P}}),\beta_{t}^{l}(\mathcal{D}_{\{1,2\}}^{k,\mathcal{P}})\}$ \\ \hline
		II-D & 
		$\beta^{l}(\mathcal{D}_{1}^{k,\mathcal{P}})+max\{\beta^{l}(\mathcal{D}_{2}^{k,\mathcal{P}}),\beta^{l}(\mathcal{D}_{\{1,2\}}^{k,\mathcal{P}})\}$ & $\beta_{t}^{l}(\mathcal{D}_{1}^{k,\mathcal{P}})+max\{\beta_{t}^{l}(\mathcal{D}_{2}^{k,\mathcal{P}}),\beta_{t}^{l}(\mathcal{D}_{\{1,2\}}^{k,\mathcal{P}})\}$ \\ \hline
		\multirow{2}{*}{II-E} & 
		\multirow{2}{*}[.5ex]{$max\{\beta^{l}(\mathcal{D}_{1}^{k,\mathcal{P}})+\beta^{l}(\mathcal{D}_{2}^{k,\mathcal{P}}),\beta^{l}(\mathcal{D}_{1}^{k,\mathcal{P}})$} & \multirow{2}{*}[.5ex]{$max\{\beta_{t}^{l}(\mathcal{D}_{1}^{k,\mathcal{P}})+\beta_{t}^{l}(\mathcal{D}_{2}^{k,\mathcal{P}}),\beta_{t}^{l}(\mathcal{D}_{1}^{k,\mathcal{P}})$} \\[5pt]
		\multirow{2}{*}{} & \multirow{2}{*}[1ex]{$+\beta^{l}(\mathcal{D}_{\{1,2\}}^{k,\mathcal{P}}),\beta^{l}(\mathcal{D}_{2}^{k,\mathcal{P}})+\beta^{l}(\mathcal{D}_{\{1,2\}}^{k,\mathcal{P}})\}$} \bigstrut[b]  &  \multirow{2}{*}[1ex]{$+\beta_{t}^{l}(\mathcal{D}_{\{1,2\}}^{k,\mathcal{P}}),\beta_{t}^{l}(\mathcal{D}_{2}^{k,\mathcal{P}})+\beta_{t}^{l}(\mathcal{D}_{\{1,2\}}^{k,\mathcal{P}})\}$} \bigstrut[b] \\[5pt] \hline
	\end{tabular}	
	\vspace{5pt}
	\caption{Summary of results for any $\mathcal{D}^k$ and $\mathcal{P}$ with fully-participated interactions between $\mathcal{D}_{1}^{k,\mathcal{P}},\mathcal{D}_{2}^{k,\mathcal{P}}$ and $\mathcal{D}_{\{1,2\}}^{k,\mathcal{P}}$}
	\label{table2}
	\vspace{-15pt}
\end{table*}

In this section, we provide the optimal linear broadcast rates $\beta_{t}^{l}(\mathcal{D},\mathcal{P})$ and $\beta^{l}(\mathcal{D},\mathcal{P})$  for all the cases of the TUICP with fully-participated interactions between the sub-digraphs $\mathcal{D}_1$, $\mathcal{D}_2$, and $\mathcal{D}_{\{1,2\}}$ of any side-information digraph $\mathcal{D}$ for any $\mathcal{P}$. For Cases I and II-A, the results are also valid for any partially-participated interactions between the sub-digraphs. 

The results of this section are summarized in Table \ref{table2}. The results marked with a $``!"$ are the ones which also hold for any  partially-participated interactions.

\subsection{CASE I}

\par In the following, we provide the optimal linear broadcast rates $\beta_{t}^{l}(\mathcal{D},\mathcal{P})$ and $\beta^{l}(\mathcal{D},\mathcal{P})$ for any two-sender problem $\mathcal{I}(\mathcal{D},\mathcal{P})$ belonging to Case I (i.e., interaction digraph belongs to Case I) with any type of interactions (fully-participated or partially-participated). The following lemmas are used to derive our results.

\begin{lem}
	For any side-information digraph $\mathcal{D}$, any message set tuple $\mathcal{P}$, and $t \geq 1$, $\beta_{t}^{l}(\mathcal{D},\mathcal{P}) \geq \beta_{t}^{l}(\mathcal{D})$.
	\label{lowerbnd}
\end{lem}
\begin{proof}
	Consider a two-sender index code with the linear broadcast rate $\beta_{t}^{l}(\mathcal{D},\mathcal{P})$. The same index code transmitted by a single-sender for the single-sender unicast ICP described by the  side-information digraph $\mathcal{D}$ satsifies the demands of all the receivers. Hence, it is a valid single-sender index code. Thus, we have the lower bound.  
\end{proof}

\begin{lem}
	For any $\mathcal{D}$, $\mathcal{P}$, and $t \geq 1$, if a side-information digraph $\mathcal{D}'$ is obtained by adding more directed edges to $\mathcal{D}$, we have  $\beta_{t}^{l}(\mathcal{D},\mathcal{P}) \geq \beta_{t}^{l}(\mathcal{D}',\mathcal{P})$, and $\beta_{t}^{l}(\mathcal{D}) \geq \beta_{t}^{l}(\mathcal{D}')$.
	\label{addedges}
\end{lem}
\begin{proof}
	Consider a linear code for the two-sender problem $\mathcal{I}(\mathcal{D},\mathcal{P})$, whose linear broadcast rate is $\beta_{t}^{l}(\mathcal{D},\mathcal{P})$. This code can be used to solve the two-sender problem $\mathcal{I}(\mathcal{D}',\mathcal{P})$, as the receivers have additional side-information including the side-information present in the original problem $\mathcal{I}(\mathcal{D},\mathcal{P})$. Hence, $\beta_{t}^{l}(\mathcal{D},\mathcal{P}) \geq \beta_{t}^{l}(\mathcal{D}',\mathcal{P})$. The same reasoning holds for the single-sender problem. Hence, $\beta_{t}^{l}(\mathcal{D}) \geq \beta_{t}^{l}(\mathcal{D}')$.
\end{proof}

We require the following lemma which is a part of Theorem 3 in \cite{tahmasbi2015critical} to derive our results. 
\begin{lem}[Theorem 3,  \cite{tahmasbi2015critical}]
	Consider a single-sender unicast index coding problem described by a side-information digraph. The set of achievable broadcast rates with linear index coding schemes does not change by removing any edge from the side-information digraph that does not lie on any directed cycle 
	\label{lemscc}
\end{lem}
We now present the main result of this sub-section.
\begin{thm}[Case I]
	For any TUICP with the side-information digraph $\mathcal{D}^k$, $k \in \{1,2,\cdots,25\}$, having any type of interactions (either fully-participated or partially-participated) between its sub-digraphs $\mathcal{D}_1^{k,\mathcal{P}}$, $\mathcal{D}_2^{k,\mathcal{P}}$, and $\mathcal{D}_{\{1,2\}}^{k,\mathcal{P}}$, for any $\mathcal{P}$, and $t$-bit messages for any $t \geq 1$, we have,
	\begin{gather*}
	(i) ~ \beta_{t}^{l}(\mathcal{D}^{k},\mathcal{P}) = \beta_{t}^{l}(\mathcal{D}_1^{k,\mathcal{P}}) + \beta_{t}^{l}(\mathcal{D}_2^{k,\mathcal{P}}) + \beta_{t}^{l}(\mathcal{D}_{\{1,2\}}^{k,\mathcal{P}}),\\
	(ii) ~ \beta^{l}(\mathcal{D}^k,\mathcal{P}) = \beta^{l}(\mathcal{D}_1^{k,\mathcal{P}}) + \beta^{l}(\mathcal{D}_2^{k,\mathcal{P}}) + \beta^{l}(\mathcal{D}_{\{1,2\}}^{k,\mathcal{P}}).
	\end{gather*}
	\label{thmonel}
\end{thm}
\begin{proof}
	First, we provide a lower bound and then provide a matching upper bound using a result of single-sender problem. 
	
	Using an optimal linear code for each $\mathcal{D}_S^{k,\mathcal{P}}$ with linear broadcast rate $\beta_{t}^{l}(\mathcal{D}_S^{k,\mathcal{P}})$, $\forall S \in \{1,2,\{1,2\}\}$, all receivers can decode their demands. Thus, we have the trivial upper bound given by
	\[
	\beta_{t}^{l}(\mathcal{D}^k,\mathcal{P}) \leq \beta_{t}^{l}(\mathcal{D}_1^{k,\mathcal{P}}) + \beta_{t}^{l}(\mathcal{D}_2^{k,\mathcal{P}}) + \beta_{t}^{l}(\mathcal{D}_{\{1,2\}}^{k,\mathcal{P}}).
	\]
	\par  Consider the single-sender problem with the same side-information digraph $\mathcal{D}$. We remove all the edges that are contributing to the interactions between the sub-digraphs $\mathcal{D}_S^{k,\mathcal{P}}$, $\forall S \in \{1,2,\{1,2\}\}$. These edges are not lying on any directed cycle as the interaction digraph is acyclic. Removing them does not alter the achievable set of linear broadcast rates for a single-sender problem according to Lemma \ref{lemscc}. Hence, using Lemma \ref{lowerbnd} we see that,
	\[  \beta_{t}^{l}(\mathcal{D}^k,\mathcal{P}) \geq \beta_{t}^{l}(\mathcal{D}^k) = \beta_{t}^{l}(\mathcal{D}_1^{k,\mathcal{P}}) + \beta_{t}^{l}(\mathcal{D}_2^{k,\mathcal{P}}) + \beta_{t}^{l}(\mathcal{D}_{\{1,2\}}^{k,\mathcal{P}}),
	\]
	which is a matching lower bound. Hence, we obtain $(i)$ in the statement of the theorem. 
	
	Taking the limit as $t\to\infty$, in the definition of $\beta^{l}(\mathcal{D}^k,\mathcal{P})$, we obtain ($ii$) in the statement of the theorem. 
\end{proof}
\subsection{CASE II-A}
\par In the following, we provide the optimal linear broadcast rates for any two-sender problem belonging to Case II-A with any type of interactions.
\begin{thm}[Case II-A]
	For any TUICP with the side-information digraph $\mathcal{D}^k$, $k \in \{26,27,28,29\}$, having any type of interactions (either fully-participated or partially-participated) between its sub-digraphs $\mathcal{D}_1^{k,\mathcal{P}}$, $\mathcal{D}_2^{k,\mathcal{P}}$, and $\mathcal{D}_{\{1,2\}}^{k,\mathcal{P}}$, for any $\mathcal{P}$, and $t$-bit messages for any $t \geq 1$, we have 
	\begin{gather*}
	(i) ~ \beta_{t}^{l}(\mathcal{D}^k,\mathcal{P}) = \beta_{t}^{l}(\mathcal{D}_1^{k,\mathcal{P}}) + \beta_{t}^{l}(\mathcal{D}_2^{k,\mathcal{P}}) + \beta_{t}^{l}(\mathcal{D}_{\{1,2\}}^{k,\mathcal{P}}),\\
	(ii) ~ \beta^{l}(\mathcal{D}^k,\mathcal{P}) = \beta^{l}(\mathcal{D}_1^{k,\mathcal{P}}) + \beta^{l}(\mathcal{D}_2^{k,\mathcal{P}}) + \beta^{l}(\mathcal{D}_{\{1,2\}}^{k,\mathcal{P}}).
	\end{gather*}
	\label{thmtwoAl}
\end{thm}
\begin{proof} 
	We consider the general form of any two-sender linear index code, and use interference alignment techniques to prove the theorem.
	
	Let ${\bf{x}}^S \in \mathbb{F}_2^{t|\mathcal{P}_S| \times 1}$ be the concatenated message vector obtained by concatenating all the messages present in $\mathcal{P}_S,$ for any non-empty set $S \subseteq \{1,2\}$. Let ${\bf{x}}^T=(({\bf{x}}^1)^T | ({\bf{x}}^2)^T | ({\bf{x}}^{\{1,2\}})^T)$, ${\bf{G}}_{i} \in \mathbb{F}_2^{l_i \times tm_i},$ $i \in \{1,2\}$, and ${\bf{G}}_{\{1,2\}}^{j} \in \mathbb{F}_2^{l_j \times tm_{\{1,2\}}},$ $j \in \{1,2,3\}$. In general, any two-sender linear index code consists of linear codes transmitted by the senders $S_i$ of the form ${\bf{G}}_{i}{\bf{x}}^i + {\bf{G}}^{i}_{\{1,2\}}{\bf{x}}^{\{1,2\}}$, $i \in \{1,2\}$, and  ${\bf{G}}_{\{1,2\}}^{3}{\bf{x}}^{\{1,2\}}$ transmitted by any one of the senders. Hence, a linear code transmitted by the senders can be written as ${\bf{G}}{\bf{x}}$ where ${\bf{G}}$ is as given in  (\ref{Geq}).  We assume that the matrices $({\bf{G}}_{1} | {\bf{G}}^{1}_{\{1,2\}}), ({\bf{G}}_{2} | {\bf{G}}^{2}_{\{1,2\}}),$ and ${\bf{G}}_{\{1,2\}}^{3}$ are full-rank matrices, which is required for the optimality of any two-sender index code.
	\begin{equation}
	{\bf{G}}=
	\left(
	\begin{array}{c|c|c} 
	{\bf{G}}_{1} & {\bf{0}}_{l_1 \times tm_2} & {\bf{G}}_{\{1,2\}}^{1} \\
	\hline
	{\bf{0}}_{l_2 \times tm_1} & {\bf{G}}_{2} & {\bf{G}}_{\{1,2\}}^{2} \\
	\hline
	{\bf{0}}_{l_3 \times tm_1} & {\bf{0}}_{l_3 \times tm_2} & {\bf{G}}_{\{1,2\}}^{3} \\
	\end{array}
	\right).
	\label{Geq}
	\end{equation}
	\par From an interference alignment perspective \cite{jafar}, the columns of ${\bf{G}}$ serve as the precoding vectors for the messages. We prove a lower bound on the number of received signal dimensions (this is same as the number of rows of ${\bf{G}}$). Note that this is also a lower bound on the length of a two-sender optimal linear index code.
	\par Observe that the receivers in $\mathcal{D}_{\{1,2\}}^{k,\mathcal{P}}$ do not have any side-information in $\mathcal{D}_{1}^{k,\mathcal{P}} \cup \mathcal{D}_{2}^{k,\mathcal{P}}$. Hence, the precoding vectors for the messages in $\mathcal{D}_{1}^{k,\mathcal{P}} \cup \mathcal{D}_{2}^{k,\mathcal{P}}$ must be independent of the precoding vectors for the messages in $\mathcal{D}_{\{1,2\}}^{k,\mathcal{P}}$. Otherwise, one or more receivers in $\mathcal{D}_{\{1,2\}}^{k,\mathcal{P}}$ can not cancel the interference caused by the precoding vectors of one or more messages in $\mathcal{D}_{1}^{k,\mathcal{P}} \cup  \mathcal{D}_{2}^{k,\mathcal{P}}$. Hence, a minimum of $t \times \beta^{l}_{t}(\mathcal{D}_{\{1,2\}}^{k,\mathcal{P}})$ dimensions (or independent vectors) are required to satisfy all the receivers (to precode the corresponding messages) in $\mathcal{D}_{\{1,2\}}^{k,\mathcal{P}}$, as it is the optimal length of any single-sender linear index code for the problem with the  side-information digraph $\mathcal{D}_{\{1,2\}}^{k,\mathcal{P}}$. Note that these $t \times \beta^{l}_{t}(\mathcal{D}_{\{1,2\}}^{k,\mathcal{P}})$ precoding vectors must be independent of those used for messages in $\mathcal{D}_{1}^{k,\mathcal{P}} \cup \mathcal{D}_{2}^{k,\mathcal{P}}$.
	\par In spite of the receivers in $\mathcal{D}_{1}^{k,\mathcal{P}}$ having some side-information in $\mathcal{D}_{2}^{k,\mathcal{P}}$ and vice-versa, the precoding vector of any message in $\mathcal{D}_{1}^{k,\mathcal{P}}$ can not be aligned with any vector in the span of the precoding vectors of the messages in $\mathcal{D}_{2}^{k,\mathcal{P}}$, due to the constraint of encoding at different senders (or equivalently due to orthogonal transmissions by the two senders), given by the first $t(m_1 + m_2)$ columns of the matrix ${\bf{G}}$. Thus, the precoding vectors of the  messages in $\mathcal{D}_{1}^{k,\mathcal{P}}$ are independent of the precoding vectors of those  in $\mathcal{D}_{2}^{k,\mathcal{P}}$. Hence, a  minimum of $t \times (\beta^{l}_{t}(\mathcal{D}_{1}^{k,\mathcal{P}})+\beta^{l}_{t}(\mathcal{D}_{2}^{k,\mathcal{P}}))$ dimensions (or independent vectors) are required to satisfy all the receivers in $\mathcal{D}_{1}^{k,\mathcal{P}}$ and $\mathcal{D}_{2}^{k,\mathcal{P}}$. 
	\par The total number of dimensions used for precoding all the messages is same as the rank of ${\bf{G}}$, which is $l_1 + l_2 +l_3$. Thus, we have the lower bound on the optimal codelength, given by 
	\[
	l_1 + l_2 + l_3 \geq  t \times (\beta_{t}^{l}(\mathcal{D}_1^{k,\mathcal{P}}) + \beta_{t}^{l}(\mathcal{D}_2^{k,\mathcal{P}}) + \beta_{t}^{l}(\mathcal{D}_{\{1,2\}}^{k,\mathcal{P}})).
	\]
	Using the single-sender index code for each sub-problem one can achieve the code length equal to $t \times \beta_{t}^{l}(\mathcal{D}_1^{k,\mathcal{P}}) + \beta_{t}^{l}(\mathcal{D}_2^{k,\mathcal{P}}) + \beta_{t}^{l}(\mathcal{D}_{\{1,2\}}^{k,\mathcal{P}})$, which gives an upper bound equal to the lower bound on the optimal codelength. Hence, we obtain $(i)$ in the statement of the theorem. 
	\par Taking the limit as $t\to\infty$, in the definition of $\beta^{l}(\mathcal{D}^k,\mathcal{P})$, we obtain (ii) in the statement of the theorem.
\end{proof}
\begin{rem}
	From Theorems \ref{thmonel} and \ref{thmtwoAl}, we observe that the side-information of receivers in $\mathcal{D}_S^{k,\mathcal{P}}$ present in other $\mathcal{D}_{S'}^{k,\mathcal{P}}$, $S,S' \in \{1,2,\{1,2\}\}$, do not help in reducing the optimal linear broadcast rates of the two-sender problems belonging to Cases I and II-A.  
\end{rem}
\subsection{CASE II-B}
\par In the following, we provide the optimal linear broadcast rates for any two-sender problem belonging to Case II-B with fully-participated  interactions. 
We state the following lower bound required to prove the main theorem of this sub-section. 
\begin{lem}
	For any TUICP, $\beta_{t}^{l}(\mathcal{D},\mathcal{P}) \geq \beta_{t}^{l}(\mathcal{D}_1) + \beta_{t}^{l}(\mathcal{D}_2)$.
	\label{lemlowbnd12}
\end{lem}
The proof follows directly from Lemma 6 in \cite{CTLO}, by considering only optimal linear broadcast rates in its proof.
\begin{thm}[Case II-B]
	For any TUICP with the  side-information digraph $\mathcal{D}^k$, $k \in \{30,31,32,33\}$, having fully-participated interactions between its sub-digraphs $\mathcal{D}_1^{k,\mathcal{P}}$, $\mathcal{D}_2^{k,\mathcal{P}}$, and $\mathcal{D}_{\{1,2\}}^{k,\mathcal{P}}$, for any $\mathcal{P}$, and $t$-bit messages for any $t \geq 1$, we have
	\begin{gather*}
	(i) ~ \beta_{t}^{l}(\mathcal{D}^k,\mathcal{P}) = max\{ \beta_{t}^{l}(\mathcal{D}_1^{k,\mathcal{P}}) + \beta_{t}^{l}(\mathcal{D}_2^{k,\mathcal{P}}), \beta_{t}^{l}(\mathcal{D}_{\{1,2\}}^{k,\mathcal{P}})\},\\
	(ii) ~ \beta^{l}(\mathcal{D}^k,\mathcal{P}) = max\{ \beta^{l}(\mathcal{D}_1^{k,\mathcal{P}}) + \beta^{l}(\mathcal{D}_2^{k,\mathcal{P}}), \beta^{l}(\mathcal{D}_{\{1,2\}}^{k,\mathcal{P}})\}.
	\end{gather*}
	\label{thmtwoA}
\end{thm}
\begin{proof}
	The proof follows from the proof of Theorem 7 in \cite{CTLO} by replacing all the optimal broadcast rates and the optimal codes present in the proof with the optimal linear broadcast rates and the  optimal linear codes respectively, and using Lemma \ref{lemlowbnd12} stated in this paper instead of Lemma 6 given in \cite{CTLO}.
\end{proof} 

\subsection{CASE II-C and CASE II-D.}
\par In the following, we provide the optimal linear  broadcast rates for any two-sender problem belonging to Cases II-C and II-D with fully-participated  interactions. We prove the following lower bound to prove the main result of this sub-section.
\begin{lem}
	For any TUICP with the side-information digraph $\mathcal{D}^k$, $k \in \{34,35,\cdots,45\}$, having fully-participated interactions between its sub-digraphs $\mathcal{D}_1^{k,\mathcal{P}}$, $\mathcal{D}_2^{k,\mathcal{P}}$, and $\mathcal{D}_{\{1,2\}}^{k,\mathcal{P}}$, for any $\mathcal{P}$, and $t$-bit messages for any $t \geq 1$, we have,  $\beta_{t}^{l}(\mathcal{D}^k,\mathcal{P}) \geq \beta_{t}^{l}(\mathcal{D}_2^{k,\mathcal{P}}) +  \beta_{t}^{l}(\mathcal{D}_{\{1,2\}}^{k,\mathcal{P}})$.
	\label{lemtwoc}
\end{lem}
\begin{proof}
	By removing the vertices belonging to $\mathcal{D}_1^{k,\mathcal{P}}$ from $\mathcal{D}^k$, we obtain a digraph $\mathcal{D}_{23}^{k,\mathcal{P}}$ which defines a TUICP. This can be considered as a single-sender ICP as both $\mathcal{P}_2$ and $\mathcal{P}_{\{1,2\}}$ are with $S_2$. Hence, we have
	\begin{equation}
	\beta_{t}^{l}(\mathcal{D}^k,\mathcal{P}) \geq \beta_{t}^{l}(\mathcal{D}_{23}^{k,\mathcal{P}}).
	\label{eqlem15}
	\end{equation}
	As there are only unidirectional edges from $\mathcal{V}(\mathcal{D}_2^{k,\mathcal{P}})$ to $\mathcal{V}(\mathcal{D}_{\{1,2\}}^{k,\mathcal{P}})$ or vice-versa (depending on the particular $k$), using Lemma \ref{lemscc}, we have 
	\begin{equation}
	\beta_{t}^{l}(\mathcal{D}_{23}^{k,\mathcal{P}}) = \beta_{t}^{l}(\mathcal{D}_2^{k,\mathcal{P}}) +  \beta_{t}^{l}(\mathcal{D}_{\{1,2\}}^{k,\mathcal{P}}).
	\label{eqlem151}
	\end{equation}
	From (\ref{eqlem15}) and (\ref{eqlem151}), we have
	\[
	\beta_{t}^{l}(\mathcal{D}^k,\mathcal{P}) \geq \beta_{t}^{l}(\mathcal{D}_2^{k,\mathcal{P}}) +  \beta_{t}^{l}(\mathcal{D}_{\{1,2\}}^{k,\mathcal{P}}).
	\]
\end{proof}
\begin{thm}[Case II-C]
	For any TUICP with the side-information digraph $\mathcal{D}^k$, $k \in \{34,35,\cdots,45\}$, having  fully-participated interactions between its sub-digraphs $\mathcal{D}_1^{k,\mathcal{P}}$, $\mathcal{D}_2^{k,\mathcal{P}}$,  and $\mathcal{D}_{\{1,2\}}^{k,\mathcal{P}}$, for any $\mathcal{P}$, and $t$-bit messages for any $t \geq 1$,we have
	\begin{gather*}
	(i) ~ \beta_{t}^{l}(\mathcal{D}^k,\mathcal{P}) = \beta_{t}^{l}(\mathcal{D}_2^{k,\mathcal{P}}) + max\{ \beta_{t}^{l}(\mathcal{D}_1^{k,\mathcal{P}}), \beta_{t}^{l}(\mathcal{D}_{\{1,2\}}^{k,\mathcal{P}})\},\\
	(ii) ~ \beta^{l}(\mathcal{D}^k,\mathcal{P}) = \beta^{l}(\mathcal{D}_2^{k,\mathcal{P}}) + max\{ \beta^{l}(\mathcal{D}_1^{k,\mathcal{P}}), \beta^{l}(\mathcal{D}_{\{1,2\}}^{k,\mathcal{P}})\}.
	\end{gather*} 
	\label{thmtwocl}
\end{thm}
\begin{proof}
	The achievability part of the proof follows from the proof of Theorem 8 given in \cite{CTLO}, by replacing all the optimal broadcast rates and optimal codes present in the proof with the optimal linear broadcast rates and optimal linear codes respectively. Hence, we have
	\begin{equation}
	\beta_{t}^{l}(\mathcal{D}^k,\mathcal{P}) \leq  \beta_{t}^{l}(\mathcal{D}_2^{k,\mathcal{P}}) + max\{ \beta_{t}^{l}(\mathcal{D}_1^{k,\mathcal{P}}), \beta_{t}^{l}(\mathcal{D}_{\{1,2\}}^{k,\mathcal{P}})\}.
	\end{equation}
	A mathcing lower bound is obtained by combining the results of Lemmas \ref{lemlowbnd12} and \ref{lemtwoc}. Hence, we have
	\begin{equation}
	\beta_{t}^{l}(\mathcal{D}^k,\mathcal{P}) \geq  \beta_{t}^{l}(\mathcal{D}_2^{k,\mathcal{P}}) + max\{ \beta_{t}^{l}(\mathcal{D}_1^{k,\mathcal{P}}), \beta_{t}^{l}(\mathcal{D}_{\{1,2\}}^{k,\mathcal{P}})\}.
	\end{equation}
	Thus, we obtain $(i)$ in the statement of the theorem. 
	\par Taking the limit as $t\to\infty$, in the definition of $\beta^{l}(\mathcal{D}^k,\mathcal{P})$, we obtain (ii) in the statement of the theorem.
\end{proof} 

We observe from Figure \ref{interenum} that the digraph $\mathcal{H}_i$, $i \in \{34,35,\cdots,45\}$, is obtained from $\mathcal{H}_j$, $j \in \{46,47,\cdots,57\}$, by interchanging the labels of vertices $1$ and $2$ of $\mathcal{H}_j$ respectively. For example,  $\mathcal{H}_{34}$ is obtained from $\mathcal{H}_{46}$. Hence, the results for Case II-D are obtained from those of Case II-C, by interchanging the labels $1$ and $2$ of the sub-digraphs in the expressions for the optimal linear broadcast rates. Thus, we state the corresponding theorem for completeness.

\begin{thm}[Case II-D]
	For any TUICP with the side-information digraph $\mathcal{D}^k$, $k \in \{46,47,\cdots,57\}$, having fully-participated  interactions  between its sub-digraphs $\mathcal{D}_1^{k,\mathcal{P}}$, $\mathcal{D}_2^{k,\mathcal{P}}$, and $\mathcal{D}_{\{1,2\}}^{k,\mathcal{P}}$, for any $\mathcal{P}$, and $t$-bit messages for any $t \geq 1$, we have
	\begin{gather*}
	(i) ~ \beta_{t}^{l}(\mathcal{D}^k,\mathcal{P}) = \beta_{t}^{l}(\mathcal{D}_1^{k,\mathcal{P}}) + max\{ \beta_{t}^{l}(\mathcal{D}_2^{k,\mathcal{P}}), \beta_{t}^{l}(\mathcal{D}_{\{1,2\}}^{k,\mathcal{P}})\},\\
	(ii) ~ \beta^{l}(\mathcal{D}^k,\mathcal{P}) = \beta^{l}(\mathcal{D}_1^{k,\mathcal{P}}) + max\{ \beta^{l}(\mathcal{D}_2^{k,\mathcal{P}}), \beta^{l}(\mathcal{D}_{\{1,2\}}^{k,\mathcal{P}})\}.
	\end{gather*}
	\label{thmtwodl}
\end{thm}

We illustrate Theorem \ref{thmtwocl} using an example.
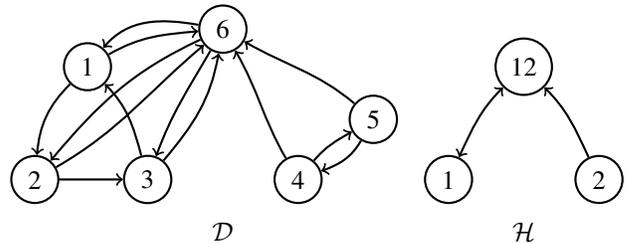
\begin{figure}[!htbp]
	\begin{center}
		\begin{tikzpicture}
		[place/.style={circle,draw=black!100,thick}]
		\node at (-5,-1.5) [place] (2c) {2};
		\node at (-3.5,-1.5) [place] (3c) {3};
		\node at (-4.3,0) [place] (1c) {1};
		\node at (-2.5,.5) [place] (6c) {6};
		\node at (-1.5,-1.5) [place] (4c) {4};
		\node at (-.5,-.7) [place] (5c) {5};
		\node at (.5,-1.5) [place] (1h) {1}; 	 	
		\node at (2.5,-1.5) [place] (2h) {2};
		\node at (1.5,0) [place] (3h) {12};	 	
		\draw (-2.5,-2.2) node {$\mathcal{D}$};
		\draw (1.5,-2.2) node {$\mathcal{H}$};
		\draw [thick,->] (1c) to [out=-135,in=85] (2c);
		\draw [thick,<->] (3h) to [out=-135,in=65] (1h);
		\draw [thick,->] (2c) to [out=0,in=180] (3c);
		\draw [thick,->] (3c) to [out=105,in=-45] (1c);
		\draw [thick,->] (4c) to [out=50,in=-150] (5c);
		\draw [thick,->] (5c) to [out=-120,in=20] (4c);
		\draw [thick,->] (4c) to [out=120,in=-60] (6c);
		\draw [thick,->] (5c) to [out=140,in=-30] (6c);
		\draw [thick,->] (6c) to [out=170,in=50] (1c);
		\draw [thick,->] (6c) to [out=-155,in=50] (2c);
		\draw [thick,->] (6c) to [out=-120,in=70] (3c);
		\draw [thick,->] (1c) to [out=30,in=-175] (6c);
		\draw [thick,->] (2c) to [out=30,in=-140] (6c);
		\draw [thick,->] (3c) to [out=50,in=-100] (6c);       		 	           		  
		\draw [thick,->] (2h) to [out=115,in=-45] (3h);	       		 	    
		\end{tikzpicture}
		\caption{Example of a two-sender problem belonging to Case II-C.}
		\label{figexamp3}
	\end{center}
\end{figure}
\begin{exmp}
	Consider a TUICP with $m=6$. Let $S_1$ and $S_2$ have $\mathcal{M}_1=\{{\bf{x}}_1,{\bf{x}}_2,{\bf{x}}_3,{\bf{x}}_6\}$ and  $\mathcal{M}_2=\{{\bf{x}}_4,{\bf{x}}_5,{\bf{x}}_6\}$ respectively. Hence, $\mathcal{P}_1=\{{\bf{x}}_1,{\bf{x}}_2,{\bf{x}}_3\}$, $\mathcal{P}_2=\{{\bf{x}}_4,{\bf{x}}_5\}$, and $\mathcal{P}_{\{1,2\}}=\{{\bf{x}}_6\}$. The side-information of all the  receivers are given as follows : $\mathcal{K}_1=\{{\bf{x}}_2,{\bf{x}}_6\},\mathcal{K}_2=\{{\bf{x}}_3,{\bf{x}}_6\},\mathcal{K}_3=\{{\bf{x}}_1,{\bf{x}}_6\},\mathcal{K}_4=\{{\bf{x}}_5,{\bf{x}}_6\},\mathcal{K}_5=\{{\bf{x}}_4,{\bf{x}}_6\},\mathcal{K}_6=\{{\bf{x}}_1,{\bf{x}}_2,{\bf{x}}_3\}$. The side-information digraph and the interaction digraph are shown in Figure \ref{figexamp3}. It is easy to verify that the associated interaction digraph is $\mathcal{H}_{37}$. Note that all the interactions are fully-participated interactions. We also know that $\mathcal{D}_1^{k,\mathcal{P}}$ and $\mathcal{D}_2^{k,\mathcal{P}}$ are cycles on vertex sets  $\{{\bf{x}}_1,{\bf{x}}_2,{\bf{x}}_3\}$ and $\{{\bf{x}}_4,{\bf{x}}_5\}$ respectively. Hence, we know using the results of \cite{SUOH}, that for any $t \geq 1$, $\beta_t(\mathcal{D}_1^{k,\mathcal{P}})=2$, $\beta_t(\mathcal{D}_2^{k,\mathcal{P}})=1$, and $\beta_t(\mathcal{D}_{\{1,2\}}^{k,\mathcal{P}})=1$. Hence, according to Theorem \ref{thmtwocl}, we have $\beta_t(\mathcal{D}^k,\mathcal{P})=1+max\{2,1\}=3$. 
	
	We provide the code for $t=1$. $S_1$ transmits ${\bf{x}}_1+{\bf{x}}_2+{\bf{x}}_6$ and ${\bf{x}}_2+{\bf{x}}_3$. $S_2$ transmits ${\bf{x}}_4+{\bf{x}}_5$. Receiver $1$ decodes ${\bf{x}}_1$ using ${\bf{x}}_1+{\bf{x}}_2+{\bf{x}}_6$ and its side-information ${\bf{x}}_2$ and ${\bf{x}}_6$. Receiver $4$ decodes ${\bf{x}}_4$ using ${\bf{x}}_4+{\bf{x}}_5$ and its side-information ${\bf{x}}_5$.  Receiver $6$ decodes ${\bf{x}}_6$ using ${\bf{x}}_1+{\bf{x}}_2+{\bf{x}}_6$ and its side-information ${\bf{x}}_1$ and ${\bf{x}}_2$. Simlarly, it can be verified that all the receivers are able to decode their demanded messages.
\end{exmp}  

From the proof of  achievability of the optimal linear broadcast rate with $t$-bit messages for any finite $t$, for a two-sender problem belonging to either Case II-C or Case II-D, we observe the following:
\begin{itemize}
	\item The addition or deletion of directed edges contributing to the existing interactions between $\mathcal{D}_1^{k,\mathcal{P}}$ and $\mathcal{D}_2^{k,\mathcal{P}}$, and to the uni-directional interaction between $\mathcal{D}_2^{k,\mathcal{P}}$ and $\mathcal{D}_{\{1,2\}}^{k,\mathcal{P}}$, do not affect the optimal linear broadcast rate of any two-sender problem belonging to Case II-C, as long as the interactions  $\mathcal{D}_1^{k,\mathcal{P}} \rightarrow \mathcal{D}_{\{1,2\}}^{k,\mathcal{P}}$ and $\mathcal{D}_{\{1,2\}}^{k,\mathcal{P}} \rightarrow \mathcal{D}_{1}^{k,\mathcal{P}}$ are fully-participated. 
	\item The addition or deletion of directed edges contributing to the existing interactions between $\mathcal{D}_1^{k,\mathcal{P}}$ and $\mathcal{D}_2^{k,\mathcal{P}}$, and to the  uni-directional interaction between $\mathcal{D}_1^{k,\mathcal{P}}$ and $\mathcal{D}_{\{1,2\}}^{k,\mathcal{P}}$, do not affect the optimal linear broadcast rate of any two-sender problem belonging to Case II-D, as long as the interactions  $\mathcal{D}_2^{k,\mathcal{P}} \rightarrow \mathcal{D}_{\{1,2\}}^{k,\mathcal{P}}$ and $\mathcal{D}_{\{1,2\}}^{k,\mathcal{P}} \rightarrow \mathcal{D}_{2}^{k,\mathcal{P}}$ are fully-participated. 
\end{itemize}

\subsection{CASE II-E}
In the following theorem, the results of Case II-E are derived using the results of Cases II-C and II-D.
\begin{thm}[Case II-E]
	For any TUICP with the  side-information digraph $\mathcal{D}^k$, $k \in \{58,59,\cdots,64\}$, having fully-participated interactions  between its sub-digraphs $\mathcal{D}_1^{k,\mathcal{P}}$, $\mathcal{D}_2^{k,\mathcal{P}}$, and $\mathcal{D}_{\{1,2\}}^{k,\mathcal{P}}$, for any $\mathcal{P}$, and $t$-bit messages for any $t \geq 1$, we have 
	\begin{gather*}
	(i) ~ \beta_{t}^{l}(\mathcal{D}^k,\mathcal{P}) =  max\{ \beta_{t}^{l}(\mathcal{D}_1^{k,\mathcal{P}}) +  \beta_{t}^{l}(\mathcal{D}_{\{1,2\}}^{k,\mathcal{P}}),  \beta_{t}^{l}(\mathcal{D}_2^{k,\mathcal{P}}) + \\  \beta_{t}^{l}(\mathcal{D}_{\{1,2\}}^{k,\mathcal{P}}), \beta_{t}^{l}(\mathcal{D}_1^{k,\mathcal{P}}) +  \beta_{t}^{l}(\mathcal{D}_2^{k,\mathcal{P}}) \}, \\
	(ii) ~ \beta^{l}(\mathcal{D}^k,\mathcal{P}) =  max\{ \beta^{l}(\mathcal{D}_1^{k,\mathcal{P}}) +  \beta^{l}(\mathcal{D}_{\{1,2\}}^{k,\mathcal{P}}), \beta^{l}(\mathcal{D}_2^{k,\mathcal{P}}) + \\ \beta^{l}(\mathcal{D}_{\{1,2\}}^{k,\mathcal{P}}), \beta^{l}(\mathcal{D}_1^{k,\mathcal{P}}) +  \beta^{l}(\mathcal{D}_2^{k,\mathcal{P}}) \}.
	\end{gather*}
	\label{thmtwoel}
\end{thm}
\begin{proof}
	We provide a lower bound using the results of Cases II-C and II-D, and then provide a matching upper bound using a code-construction for the two-sender problem, by utilising optimal linear single-sender codes of the sub-problems. 
	
	Given any side-information digraph $\mathcal{D}^k$, such that $k \in \{58,59,\cdots,64\}$, with fully-participated interactions  between its sub-digraphs $\mathcal{D}_{1}^{k,\mathcal{P}}$, $\mathcal{D}_{2}^{k,\mathcal{P}}$, and $\mathcal{D}_{\{1,2\}}^{k,\mathcal{P}}$, we obtain $(i)$ one of the side-information digraphs $\mathcal{D}^{r'}, r' \in \{44,45\}$, and $(ii)$ one of the side-information digraphs $\mathcal{D}^{r''}, r'' \in \{56,57\}$, with the same $\mathcal{D}_{1}^{k,\mathcal{P}}$, $\mathcal{D}_{2}^{k,\mathcal{P}}$, and $\mathcal{D}_{\{1,2\}}^{k,\mathcal{P}}$ having fully-participated interactions, by adding appropriate edges between the sub-digraphs of $\mathcal{D}^k$. From Lemma \ref{addedges}, we have
	\begin{equation}
	\beta_{t}^l(\mathcal{D}^k,\mathcal{P}) \geq  \beta_{t}^l(\mathcal{D}^{r'},\mathcal{P}),
	\label{case2e_lbnd1}
	\end{equation}
	\begin{equation}
	\beta_{t}^l(\mathcal{D}^k,\mathcal{P}) \geq  \beta_{t}^l(\mathcal{D}^{r''},\mathcal{P})
	\label{case2e_lbnd2}.
	\end{equation}
	\par Combining the results of Theorem \ref{thmtwocl} and Theorem \ref{thmtwodl} using (\ref{case2e_lbnd1}) and (\ref{case2e_lbnd2}), we have
	\begin{equation}
	\begin{split}
	& \beta_{t}^l(\mathcal{D}^k,\mathcal{P}) \geq max\{\beta_{t}^l(\mathcal{D}_1^{k,\mathcal{P}})+\beta_{t}^l(\mathcal{D}_{\{1,2\}}^{k,\mathcal{P}}),\\ & \beta_{t}^l(\mathcal{D}_2^{k,\mathcal{P}})+\beta_{t}^l(\mathcal{D}_{\{1,2\}}^{k,\mathcal{P}}),\beta_{t}^l(\mathcal{D}_1^{k,\mathcal{P}})+\beta_{t}^l(\mathcal{D}_2^{k,\mathcal{P}})\}.
	\end{split}
	\end{equation}
	\par We provide an upper bound by giving a code construction. When $\beta_{t}^l(\mathcal{D}_{\{1,2\}}^{k,\mathcal{P}}) \leq min\{\beta_{t}^l(\mathcal{D}_1^{k,\mathcal{P}}),\beta_{t}^l(\mathcal{D}_2^{k,\mathcal{P}})\}$, the achievability follows from the proof of Theorem $9$ in \cite{CTLO}, by replacing all the optimal broadcast rates and optimal codes present in the proof with optimal linear broadcast rates and optimal linear codes respectively.
	\par Without loss of generality, we can assume that we have $\beta_{t}^l(\mathcal{D}_2^{k,\mathcal{P}}) \leq min\{\beta_{t}^l(\mathcal{D}_1^{k,\mathcal{P}}),\beta_{t}^l(\mathcal{D}_{\{1,2\}}^{k,\mathcal{P}})\}$. The case with $\beta_{t}^l(\mathcal{D}_1^{k,\mathcal{P}}) \leq min\{\beta_{t}^l(\mathcal{D}_2^{k,\mathcal{P}}),\beta_{t}^l(\mathcal{D}_{\{1,2\}}^{k,\mathcal{P}})\}$ can be proved similarly. Let $\mathcal{C}_S$ be a linear code with the optimal linear broadcast rate $\beta_{t}^l(\mathcal{D}_S^{k,\mathcal{P}})$ for the single-sender unicast ICP described by $\mathcal{D}_S^{k,\mathcal{P}}$. Our code for the original TUICP $\mathcal{I}(\mathcal{D}^k,\mathcal{P})$ is given as follows:
	\[  \mathcal{C}_1 \oplus \mathcal{C}_{\{1,2\}}[1:t\beta_{t}^l(\mathcal{D}_2^{k,\mathcal{P}})]  ~ ~ ~\mbox{sent by $S_1$}, \]
	\[  \mathcal{C}_2 \oplus \mathcal{C}_{\{1,2\}}[1:t\beta_{t}^l(\mathcal{D}_2^{k,\mathcal{P}})]  ~ ~ ~ \mbox{sent by $S_2$}, \]
	and 
	$\mathcal{C}_{\{1,2\}}[1+t\beta_{t}^l(\mathcal{D}_2^{k,\mathcal{P}}):t\beta_{t}^l(\mathcal{D}_{\{1,2\}}^{k,\mathcal{P}})]$ sent by any one of $S_1$ or $S_2$.    
	
	The overall length of the two-sender code is given by
	\begin{gather*} t(\beta_{t}^l(\mathcal{D}_1^{k,\mathcal{P}})+\beta_{t}^l(\mathcal{D}_2^{k,\mathcal{P}})+(\beta_{t}^l(\mathcal{D}_{\{1,2\}}^{k,\mathcal{P}})-\beta_{t}^l(\mathcal{D}_2^{k,\mathcal{P}}))) \\ =t(\beta_{t}^l(\mathcal{D}_1^{k,\mathcal{P}})+\beta_{t}^l(\mathcal{D}_{\{1,2\}}^{k,\mathcal{P}})),
	\end{gather*}
	with linear broadcast rate $\beta_{t}^l(\mathcal{D}_1^{k,\mathcal{P}})+\beta_{t}^l(\mathcal{D}_{\{1,2\}}^{k,\mathcal{P}})$. 
	\par We provide the decoding procedure for receivers in the side-information digraphs $\mathcal{D}^k$ with $k \in \{58,59,\cdots,62\}$. The decoding procedure for those in the side-information digraphs $\mathcal{D}^k$ with $k \in \{63,64\}$ is similar. Receivers belonging to $\mathcal{D}_1^{k,\mathcal{P}}$ and $\mathcal{D}_2^{k,\mathcal{P}}$ recover their demanded messages using
	\begin{gather*}
	(\mathcal{C}_2 \oplus \mathcal{C}_{\{1,2\}}[1:t\beta_{t}^l(\mathcal{D}_2^{k,\mathcal{P}})])  \oplus (\mathcal{C}_1 \oplus \mathcal{C}_{\{1,2\}}[1:t\beta_{t}^l(\mathcal{D}_2^{k,\mathcal{P}})])\\
	= \mathcal{C}_1 \oplus \mathcal{C}_2,
	\end{gather*}
	and their side-information $\mathcal{P}_2$ and $\mathcal{P}_1$ respectively. Receivers belonging to $\mathcal{D}_{\{1,2\}}^{k,\mathcal{P}}$ recover their demanded messages using $\mathcal{C}_{\{1,2\}}[t\beta_{t}^l(\mathcal{D}_2^{k,\mathcal{P}})+1:t\beta_{t}^l(\mathcal{D}_{\{1,2\}}^{k,\mathcal{P}})]$ and either  $\mathcal{C}_2 \oplus \mathcal{C}_{\{1,2\}}[1:t\beta_{t}^l(\mathcal{D}_2^{k,\mathcal{P}})]$ or $\mathcal{C}_1 \oplus \mathcal{C}_{\{1,2\}}[1:t\beta_{t}^l(\mathcal{D}_2^{k,\mathcal{P}})]$, and their side-information, depending on the presence of the interaction  $\mathcal{D}_{\{1,2\}}^{k,\mathcal{P}} \rightarrow \mathcal{D}_{2}^{k,\mathcal{P}}$ or $\mathcal{D}_{\{1,2\}}^{k,\mathcal{P}} \rightarrow \mathcal{D}_{1}^{k,\mathcal{P}}$ respectively. This code construction provides an upper bound which meets the lower bound, and thus (i) in the statement of the theorem is proved.
	\par Taking the limit as $t\to\infty$ in the definition of $\beta^l(\mathcal{D}^k,\mathcal{P})$, we obtain (ii) in the statement of the theorem.
\end{proof}

We illustrate the theorem using an example.
\begin{figure}[!htbp]
	\begin{center}         		
		\begin{tikzpicture}
		[place/.style={circle,draw=black!100,thick}]
		\node at (-4.0,-1.7) [place] (1c) {1};
		\node at (-5.5,-1) [place] (2c) {2};
		\node at (-5,.6) [place] (3c) {3};
		\node at (-1.1,-.3) [place] (6c) {6};
		\node at (-3.5,1.2) [place] (4c) {4};
		\node at (-2.0,1.1) [place] (5c) {5};
		\node at (-2.0,-1.5) [place] (7c) {7};
		\node at (.5,-1.5) [place] (1h) {1}; 	 	
		\node at (2.5,-1.5) [place] (2h) {2};
		\node at (1.5,0) [place] (3h) {12};	 	
		\draw (-3.2,-2.2) node {$\mathcal{D}$};
		\draw (1.5,-2.2) node {$\mathcal{H}$};
		\draw [thick,->] (1c) to [out=-170,in=-55] (2c); 
		\draw [thick,->] (2c) to [out=-35,in=165] (1c);
		\draw [thick,->] (2c) to [out=90,in=-145] (3c);
		\draw [thick,->] (3c) to [out=-120,in=70] (2c);
		\draw [thick,->] (1c) to [out=130,in=-80] (3c);
		\draw [thick,->] (3c) to [out=-60,in=110] (1c);
		\draw [thick,->] (7c) to [out=70,in=-140] (6c);
		\draw [thick,->] (6c) to [out=-110,in=40] (7c);
		\draw [thick,->] (4c) to [out=10,in=160] (5c);
		\draw [thick,->] (5c) to [out=-170,in=-15] (4c);
		\draw [thick,->] (1c) to [out=85,in=-120] (4c);
		\draw [thick,->] (1c) to [out=70,in=-130] (5c);
		\draw [thick,->] (2c) to [out=40,in=-150] (5c);
		\draw [thick,->] (2c) to [out=50,in=-150] (4c);      		
		\draw [thick,->] (3c) to [out=60,in=140] (5c);
		\draw [thick,->] (3c) to [out=50,in=-180] (4c);       		
		\draw [thick,->] (4c) to [out=-30,in=140] (6c);
		\draw [thick,->] (4c) to [out=-50,in=110] (7c);      		
		\draw [thick,->] (5c) to [out=-40,in=110] (6c);
		\draw [thick,->] (5c) to [out=-90,in=90] (7c); 	      	   
		\draw [thick,->] (6c) to [out=-160,in=10] (1c);
		\draw [thick,->] (6c) to [out=-180,in=10]  (2c);
		\draw [thick,->] (6c) to [out=160,in=-20]  (3c); 
		\draw [thick,->] (7c) to [out=-170,in=-10] (1c);
		\draw [thick,->] (7c) to [out=170,in=-10]  (2c);
		\draw [thick,->] (7c) to [out=140,in=-40]  (3c);
		\draw [thick,->] (1h) to [out=55,in=-125] (3h);
		\draw [thick,->] (2h) to [out=-180,in=0] (1h); 
		\draw [thick,->] (3h) to [out=-55,in=125] (2h);
		\end{tikzpicture}
		\caption{Example of a two-sender problem belonging to Case II-E.}
		\label{figexamp4}
	\end{center}
\end{figure}
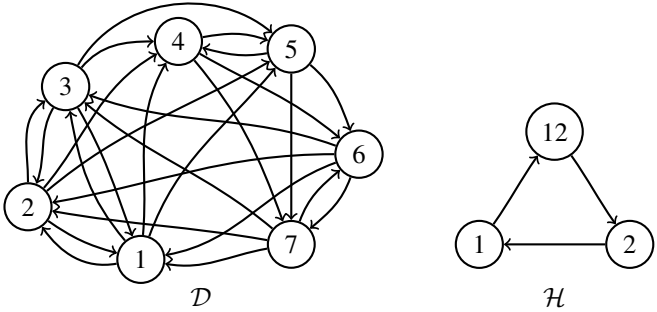
\begin{exmp}
	Consider a TUICP with $m=7$. Let $S_1$ have the message set $\mathcal{M}_1=\{{\bf{x}}_1,{\bf{x}}_2,{\bf{x}}_3,{\bf{x}}_4,{\bf{x}}_5\}$ and $S_2$ have the message set   $\mathcal{M}_2=\{{\bf{x}}_4,{\bf{x}}_5,{\bf{x}}_6,{\bf{x}}_7\}$. Hence, $\mathcal{P}_1=\{{\bf{x}}_1,{\bf{x}}_2,{\bf{x}}_3\}$, $\mathcal{P}_2=\{{\bf{x}}_6,{\bf{x}}_7\}$, and $\mathcal{P}_{\{1,2\}}=\{{\bf{x}}_4,{\bf{x}}_5\}$. The side-information of all the receivers are given as follows: 
	\[
	\mathcal{K}_1=\{{\bf{x}}_2,{\bf{x}}_3,{\bf{x}}_4,{\bf{x}}_5\},\mathcal{K}_2=\{{\bf{x}}_3,{\bf{x}}_1,{\bf{x}}_4,{\bf{x}}_5\},
	\]
	\[\mathcal{K}_3=\{{\bf{x}}_1,{\bf{x}}_2,{\bf{x}}_4,{\bf{x}}_5\},\mathcal{K}_4=\{{\bf{x}}_5,{\bf{x}}_6,{\bf{x}}_7\},
	\]
	\[
	\mathcal{K}_5=\{{\bf{x}}_4,{\bf{x}}_6,{\bf{x}}_7\},\mathcal{K}_6=\{{\bf{x}}_7,{\bf{x}}_1,{\bf{x}}_2,{\bf{x}}_3\},\]
	\[
	\mathcal{K}_7=\{{\bf{x}}_6,{\bf{x}}_1,{\bf{x}}_2,{\bf{x}}_3\}.
	\]
	The side-information digraph and the interaction digraph are given in Figure \ref{figexamp4}. It is easy to verify that the interaction digraph $\mathcal{H}$ of the side-information digraph $\mathcal{D}$ is $\mathcal{H}_{64}$. Note that all the interactions are fully-participated interactions. We observe that $\mathcal{D}_1^{64,\mathcal{P}}$ is a clique, and that  $\mathcal{D}_2^{64,\mathcal{P}}$ and  $\mathcal{D}_{\{1,2\}}^{64,\mathcal{P}}$ are cycles. 
	The optimal linear  broadcast rate of a clique is 1. Hence, from the results of \cite{SUOH} and the above fact, for any $t \geq 1$ we have, $\beta_t(\mathcal{D}_1^{64,\mathcal{P}})=1$, $\beta_t(\mathcal{D}_2^{64,\mathcal{P}})=1$, and  $\beta_t(\mathcal{D}_{\{1,2\}}^{64,\mathcal{P}})=1$. Hence, according to Theorem \ref{thmtwoel}, we have $\beta_t(\mathcal{D}^{64},\mathcal{P})=max\{1+1,1+1,1+1\}=2$. We provide the code for $t=1$. $S_1$ transmits ${\bf{x}}_1+{\bf{x}}_2+{\bf{x}}_3+{\bf{x}}_4+{\bf{x}}_5$. $S_2$  transmits ${\bf{x}}_4+{\bf{x}}_5+{\bf{x}}_6+{\bf{x}}_7$. Receiver $1$ decodes ${\bf{x}}_1$ using ${\bf{x}}_1+{\bf{x}}_2+{\bf{x}}_3+{\bf{x}}_4+{\bf{x}}_5$ and its side-information  $\mathcal{K}_1=\{{\bf{x}}_2,{\bf{x}}_3,{\bf{x}}_4,{\bf{x}}_5\}$. Reciever $4$ decodes ${\bf{x}}_4$ using ${\bf{x}}_4+{\bf{x}}_5+{\bf{x}}_6+{\bf{x}}_7$ and its side-information $\mathcal{K}_4=\{{\bf{x}}_5,{\bf{x}}_6,{\bf{x}}_7\}$. Receiver $6$ decodes ${\bf{x}}_6$ using $ ({\bf{x}}_1+{\bf{x}}_2+{\bf{x}}_3+{\bf{x}}_4+{\bf{x}}_5) \oplus ({\bf{x}}_4+{\bf{x}}_5+{\bf{x}}_6+{\bf{x}}_7) = {\bf{x}}_1+{\bf{x}}_2+{\bf{x}}_3+{\bf{x}}_6+{\bf{x}}_7$ and its side-information $\mathcal{K}_6=\{{\bf{x}}_7,{\bf{x}}_1,{\bf{x}}_2,{\bf{x}}_3\}$. Similarly, other receivers can  decode their demanded messages.
\end{exmp}

The following corollary emphasizes the significance of finding the optimal linear broadcast rates $\beta_{t}^{l}(\mathcal{D},\mathcal{P})$ and $\beta^{l}(\mathcal{D},\mathcal{P})$ of two-sender problems with fully-participated interactions. The proof follows directly form Lemma \ref{addedges}.

\begin{cor}
	The optimal linear broadcast rates $\beta_{t}^{l}(\mathcal{D},\mathcal{P})$ and $\beta^{l}(\mathcal{D},\mathcal{P})$ given in one of the  Theorems \ref{thmtwoA}, \ref{thmtwocl},  \ref{thmtwodl}, and \ref{thmtwoel} for a two-sender problem with fully-participated interactions serve as lower bounds for the corresponding optimal linear broadcast rates of a two-sender problem belonging to the respective case (Cases II-B, II-C, II-D, and II-E respectively) with partially-participated interactions, if the three sub-digraphs of the problem with partially-participated interactions are same as those of the corresponding problem  with fully-participated interactions respectively.
\end{cor}

\section{Extension of the results to the multi-sender unicast ICP}
In this section, we show that the results and the proof techniques used in the previous section can be used to obtain the optimal linear broadcast rates for some cases of the multi-sender ICP. When it is difficult to analyze any  multi-sender problem to obtain the optimal linear broadcast rate, a sub-optimal linear broadcast rate can be found by partitioning the problem into smaller multi-sender problems for which optimal results are known. We extend the definitions and notations given in Section II for the multi-sender unicast ICP. 

\subsection{Extension of the definitions for the multi-sender unicast index coding problem}
Consider an $s$-sender unicast ICP. The $i$th sender has the message set given by $\mathcal{M}_i$, $i \in [s]$, where $\mathcal{M}_i \subset \mathcal{M}$. We partition the set of all messages given by $\mathcal{M} =\{{\bf{x}}_1,{\bf{x}}_2,\cdots,{\bf{x}}_{m}\}$ into $2^s-1$ non-empty sets. We define the set of messages $\mathcal{P}_{S}$, for all non-empty $S \subseteq [s]$ as follows:
\begin{gather*}
\mathcal{P}_{S} = \{{\bf{x}}_j: {\bf{x}}_j \in \mathcal{M}_i, j \in [m], \forall i \in S 
~\mbox{and}~ \\ {\bf{x}}_j \notin \mathcal{M}_{i'}, \forall i' \in [s] \setminus S \}.
\end{gather*}
Hence, $\mathcal{P}_{S}$ is the set of messages possessed by every sender represented by the elements of $S$ and possessed by no other senders. Without loss of generality, we assume that $\mathcal{P}_S \neq \Phi$, for any $S$ such that $|S|=1$. If $\mathcal{P}_S = \Phi$, such that $|S|=1$, then all the messages with the sender represented by $S$ are also present with another sender and hence the sender with $\mathcal{P}_S=\Phi$ can be  removed from the problem setup. Let $\mathcal{D}_{S}$ be the sub-digraph of the side-information digraph $\mathcal{D}$, induced by the vertices $\{v_j : {\bf{x}}_j \in \mathcal{P}_{S}, j \in [m]\}$, where $S \subseteq [s]$ is non-empty. There are $2^s-1$ non-empty subsets of the set $[s]$ and hence a maximum of $2^s-1$ disjoint sub-digraphs $\mathcal{D}_{S}$ if all the sets $\mathcal{P}_{S}$, for non-empty $S \subseteq [s]$, are non-empty. When $\mathcal{P}_{S}=\Phi$, the corresponding $\mathcal{D}_{S}$ does not exist. 

Consider the interaction digraph $\mathcal{H}$ corresponding to the side-information digraph $\mathcal{D}$ for the $s$-sender problem, where  $\mathcal{V}(\mathcal{H}) = \{S : \mathcal{P}_{S} \neq \Phi, S \subseteq [s],  S \neq \Phi \}$, and $\mathcal{E}(\mathcal{H})=\{(S',S'') | \mathcal{D}_{S'} \rightarrow \mathcal{D}_{S''}, S',S'' \in \mathcal{V}(\mathcal{H})\}$. Let  $f: \mathcal{V}(\mathcal{D}) \rightarrow \mathcal{V}(\mathcal{H})$ be a function such that  $f(v)=S$ if $v \in \mathcal{V}(\mathcal{D}_{S})$. All other definitions given in Section II can be similarly extended to the multi-sender unicast ICP.  If $\mathcal{P}_{S} = \Phi$, we set $ \beta^{l}_t(\mathcal{D}_S) = 0$.

We illustrate all the  definitions, with a running example.

\begin{exmp}
	Consider a three-sender problem with $m=8$ messages and the message sets at the senders given by   $\mathcal{M}_1=\{{\bf{x}}_1,{\bf{x}}_4,{\bf{x}}_7,{\bf{x}}_8\}$, $\mathcal{M}_2=\{{\bf{x}}_2,{\bf{x}}_4,{\bf{x}}_5,{\bf{x}}_6,{\bf{x}}_7,{\bf{x}}_8\}$, and $\mathcal{M}_3=\{{\bf{x}}_3,{\bf{x}}_5,{\bf{x}}_6,{\bf{x}}_7\}$. The $i$th receiver demands ${\bf{x}}_i$, $\forall i \in [m]$. The  side-information of the receivers are given as follows:
	\begin{gather*}
	\mathcal{K}_1 = \{{\bf{x}}_2,{\bf{x}}_5,{\bf{x}}_7\}, \mathcal{K}_2 = \{{\bf{x}}_1,{\bf{x}}_3,{\bf{x}}_5,{\bf{x}}_7\},\\ \mathcal{K}_3 = \{{\bf{x}}_2,{\bf{x}}_6\}, \mathcal{K}_{4} = \{{\bf{x}}_6,{\bf{x}}_8\},  \mathcal{K}_{5}=\{{\bf{x}}_6\},\\ \mathcal{K}_{6}=\{{\bf{x}}_5\},\mathcal{K}_{7}=\{{\bf{x}}_4\},\mathcal{K}_{8}=\{{\bf{x}}_4\}.
	\end{gather*}
	The side-information digraph is shown in Figure \ref{fig3ICP0}.
	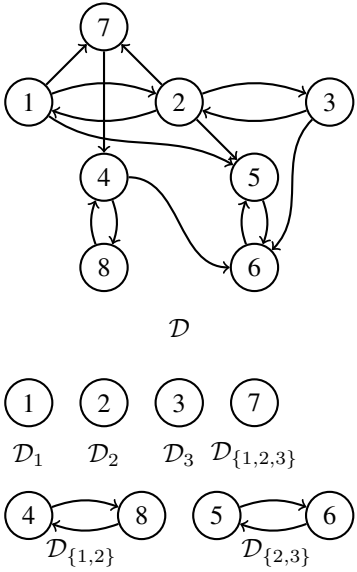
\begin{figure}
		\centering
		\begin{tikzpicture}
		[place/.style={circle,draw=black!100,thick}]
		\node at (-3,0) [place] (1c) {1};
		\node at (-1,0) [place] (2c) {2};
		\node at (1,0) [place] (3c) {3};
		\node at (-2,-1) [place] (4c) {4};
		\node at (-2,-2.2) [place] (8c) {8}; 	 	
		\node at (0,-1) [place] (5c) {5};
		\node at (0,-2.2) [place] (6c) {6};	 	
		\draw (-1,-3) node {$\mathcal{D}$};
		\node at (-2,1) [place] (7c) {7};
		\draw [thick,->] (1c) to [out=20,in=160] (2c);
		\draw [thick,->] (2c) to [out=-160,in=-20] (1c);
		\draw [thick,->] (2c) to [out=20,in=160] (3c);
		\draw [thick,->] (3c) to [out=-160,in=-20] (2c);
		\draw [thick,->] (4c) to [out=-70,in=70] (8c);
		\draw [thick,->] (8c) to [out=110,in=-110] (4c);
		\draw [thick,->] (5c) to [out=-70,in=70] (6c);
		\draw [thick,->] (6c) to [out=110,in=-110] (5c);
		\draw [thick,->] (1c) to [out=45,in=-135] (7c);
		\draw [thick,->] (2c) to [out=135,in=-45] (7c);
		\draw [thick,->] (2c) to [out=-45,in=135] (5c);
		\draw [thick,->] (4c) to [out=0,in=180]  (6c);
		\draw [thick,->] (7c) to [out=-90,in=90] (4c);
		\draw [thick,->] (1c) to [out=-35,in=150] (5c);
		\draw [thick,->] (3c) to [out=-135,in=45] (6c);
		\node at (-3,-4) [place] (1'c) {1};
		\draw (-3,-4.7) node {$\mathcal{D}_1$};
		\node at (-2,-4)  [place] (2'c) {2};
		\draw (-2,-4.7) node {$\mathcal{D}_2$};
		\node at (-1,-4)  [place] (3'c) {3};	
		\draw (-1,-4.7) node {$\mathcal{D}_3$};
		\node at (0,-4)  [place] (123'c) {7};	
		\draw (0,-4.7) node {$\mathcal{D}_{\{1,2,3\}}$};
		\node at (-3,-5.5) [place] (124'c) {4};
		\node at (-1.5,-5.5) [place] (128'c) {8};
		\draw [thick,->] (124'c) to [out=20,in=160] (128'c);
		\draw [thick,->] (128'c) to [out=-160,in=-20] (124'c);
		\draw (-2.3,-6) node {$\mathcal{D}_{\{1,2\}}$};
		\node at (-.5,-5.5) [place] (235'c) {5};
		\node at (1,-5.5) [place] (236'c) {6};	 
		\draw [thick,->] (235'c) to [out=20,in=160] (236'c);
		\draw [thick,->] (236'c) to [out=-160,in=-20] (235'c);
		\draw (.3,-6) node {$\mathcal{D}_{\{2,3\}}$}; 	 		
		\end{tikzpicture}
		\caption{The side-information digraph and the  sub-digraphs of Example \ref{exmp40}.}
		\label{fig3ICP0}
	\end{figure}
	From the definition of $\mathcal{P}_{S}$ for any non-empty $S \subseteq [s]$, we have the following:
	\begin{gather*}
	\mathcal{P}_1 = \{{\bf{x}}_1\}, \mathcal{P}_2 = \{{\bf{x}}_2\}, \mathcal{P}_3 = \{{\bf{x}}_3\}, \mathcal{P}_{\{1,2\}} = \{{\bf{x}}_4,{\bf{x}}_8\}, \\ \mathcal{P}_{\{2,3\}}=\{{\bf{x}}_5,{\bf{x}}_6\}, \mathcal{P}_{\{1,3\}}=\Phi,\mathcal{P}_{\{1,2,3\}}=\{{\bf{x}}_7\},\\
	\mathcal{P}=(\mathcal{P}_1,\mathcal{P}_2,\mathcal{P}_3,\mathcal{P}_{\{1,2\}},\mathcal{P}_{\{2,3\}},\mathcal{P}_{\{1,2,3\}}).
	\end{gather*}
	The sub-digraphs $\mathcal{D}_{S}$ induced by $\mathcal{P}_{S} \neq \Phi$, for all non-empty $S \subseteq [s]$, are also shown in Figure \ref{fig3ICP0}. The interaction digraph is shown in  Figure \ref{fig3ICP}.
	\label{exmp40}
\end{exmp}

We know that there are $2^s-1$ disjoint sub-digraphs for an $s$-sender problem, if $\mathcal{P}_{S} \neq \Phi$ for all non-empty $S \subseteq [s]$. A maximum of $2^s-1 \choose 2$ undirected edges are possible in $\mathcal{H}$ and hence a maximum of $2{2^s-1 \choose 2}$ directed edges are possible. Hence, there are $2^{2 {{2^s-1}\choose{2}}}$ possible interaction digraphs when $\mathcal{P}_{S} \neq \Phi$ for all non-empty $S \subseteq [s]$. For example, if $s=3$, there are $2^{42}$ interaction digraphs, when $\mathcal{P}_{S} \neq \Phi$ for all non-empty $S \subseteq [s]$. Note that for $s=2$, there are 64 interaction digraphs. This double-exponential increase (with the number of senders) in the number of interaction digraphs makes it difficult to completely classify the interaction digraphs into different classes. Thus, it is difficult to enumerate and classify the multi-sender unicast ICP with $s$ senders, for $s \geq 3$, based on the interaction digraphs, as done for the two-sender problem.

\subsection{Extensions of Theorems \ref{thmonel} and \ref{thmtwoAl} to multi-sender ICP}
In the following, we state an extension of Theorem \ref{thmonel} to multi-sender unicast ICPs. The proof is not provided as it follows on the same lines as that of Theorem \ref{thmonel}. 

\begin{thm}
	For any $s$-sender unicast ICP, $s \geq 2$, with the side-information digraph $\mathcal{D}$, with any type of interactions (fully-participated or partially-participated) between the sub-digraphs $\mathcal{D}_{S}$, for non-empty sets $S \subseteq [s]$, for any $\mathcal{P}$, and with the interaction digraph given by  $\mathcal{H}$ being acyclic, we have 
	\begin{gather*}
	(i) ~ \beta_{t}^{l}(\mathcal{D},\mathcal{P}) = \underset{S : S \subseteq [s]}{\Sigma} \beta_t^{l}(\mathcal{D}_{S}),\\
	(ii) ~ \beta^{l}(\mathcal{D},\mathcal{P}) = \underset{S : S \subseteq [s]}{\Sigma} \beta^{l}(\mathcal{D}_S).
	\end{gather*}
\end{thm}

We next provide a result which uses the proof technique used in Theorem \ref{thmtwoAl}, along with an additional fact from graph theory, to provide the optimal linear broadcast rate for a class of $s$-sender unicast ICP, for any $s \geq 2$. We need the following lemma to prove our results.
\begin{lem}[\cite{DEK}]
	Any acyclic digraph has at least one topological ordering. 
\end{lem}
The definition of topological ordering is given in Section II. This result implies that there exists a vertex with no incoming edges and another vertex with no outgoing edges in any acyclic digraph. The following theorem extends Theorem \ref{thmtwoAl} to any $s$-sender unicast ICP, $s \geq 3$.
\begin{thm}
	For any $s$-sender unicast ICP, $s \geq 2$, with the side-information digraph $\mathcal{D}$ and the interaction digraph $\mathcal{H}$, with any type of interactions (fully-participated or partially-participated), let the interactions between the sub-digraphs satisfy all the following conditions. 
	\par $(i)$ If $|S|=1$ and $|S'|>1$, $S,S' \subseteq [s]$, and there exists an interaction between $\mathcal{D}_S$ and $\mathcal{D}_{S'}$, it must only be of the form $\mathcal{D}_S \rightarrow \mathcal{D}_{S'}$. 
	\par $(ii)$ The sub-digraph induced by all the vertices $\{S : S \in \mathcal{V}(\mathcal{H}), |S| > 1\}$, must be acyclic.
	\par $(iii)$ If $|S|=1$ and $|S'|=1$, $S,S' \subset [s]$, any type of interactions are allowed. \\
	
	Then we have, 
	\begin{gather*}
	(i) ~ \beta_{t}^{l}(\mathcal{D},\mathcal{P}) = \underset{S : S \subseteq [s]}{\Sigma} \beta_t^{l}(\mathcal{D}_S),\\
	(ii) ~ \beta^{l}(\mathcal{D},\mathcal{P}) = \underset{S : S \subseteq [s]}{\Sigma} \beta^{l}(\mathcal{D}_S).
	\end{gather*}
	\label{thm3ICP2}
\end{thm}
\begin{proof}
	The proof follows on the same lines as that of Theorem \ref{thmtwoAl}. First, we prove a lower bound on the minimum length of the index code or equivalently the minimum number of received signal dimensions needed for all the receivers to decode their demands, using an   interference alignment perspective. The matching upper bound follows on the same lines as that given in the proof of Theorem \ref{thmtwoAl}. Throughout the proof we consider $S$ such that $S \subseteq [s]$ is non-empty. 
	
	Due to the constraints of encoding or equivalently orthogonal transmissions in time by different senders (as seen in the proof of Theorem \ref{thmtwoAl}), we see that at least $d_{|S|=1} = t \times \underset{S : |S|=1, S \subset [s]}{\Sigma} \beta_t^{l}(\mathcal{D}_S)$ dimensions (independent vectors) are required for precoding all the messages $\mathcal{P}_{S}$ such that $|S|=1$.
	
	Due to the condition $(i)$ in the statement of the theorem, the precoding vector of any message belonging to any $\mathcal{P}_{S}$ such that  $|S|>1, S \subseteq [s]$, can not be aligned with any vector in the span of the precoding vectors of messages belonging to any  $\mathcal{P}_{S'}$ such that $|S'|=1$. This is because no receiver belonging to any $\mathcal{D}_S$ with $|S|>1$ has any side-information in any $\mathcal{D}_{S'}$ such that $|S'|=1$. Hence, 
	if $d_{|S|>1}$ is the minimum number of  dimensions required for the precoding vectors of all $\mathcal{P}_S$ such that $|S|>1$, then we need at least $d_{|S|=1}+d_{|S|>1}$ number of dimensions in total, which is also a lower bound on the length of the index code.
	
	Now, we calculate the minimum value of $d_{|S|>1}$. Using the topological ordering we can order all ${S}$ with $|S| > 1$ in $\mathcal{H}$.
	We obtain a ${S'}$ which has no outgoing edges (i.e., vertices in $\mathcal{D}_{S'}$ have no side-information in other sub-digraphs). Hence, according to the interference alignment principle, the precoding vectors of the messages in $\mathcal{D}_{S'}$ must be independent of the precoding vectors of the messages in other $\mathcal{D}_{S}$ with  $|S| > 1$. Hence, we see that at least $\beta_t^{l}(\mathcal{D}_{S'})$ dimensions are required for precoding all the messages in  $\mathcal{D}_{S'}$. Deleting ${S'}$ and its incoming edges from the topological ordering, we get another topological ordering, as the remaining digraph in $\mathcal{H}$ induced by ${S}$ with $|S| > 1$ is still acyclic. Using the same argument for  ${S''}$ which is the new vertex with no outgoing edges, we see that the precoding vectors of the messages in $\mathcal{D}_{S''}$ must be independent of those belonging to $\mathcal{D}_{S}$ with  $|S| > 1$. Hence, at least $\beta_t^{l}(\mathcal{D}_{S''})$ dimensions are required for precoding all the messages in  $\mathcal{D}_{S''}$. Now ${S''}$ and all its incoming edges are deleted from the topological ordering.   Repeating this argument until all the vertices in the topological ordering are deleted, we see that  $d_{|S|>1} = t \times \underset{S : |S|>1, S \subseteq [s]}{\Sigma} \beta_t^{l}(\mathcal{D}_S)$. Hence, the lower bound given by 
	\begin{gather*}
	\beta_{t}^{l}(\mathcal{D},\mathcal{P}) \geq  \underset{S : S \subseteq [s]}{\Sigma} \beta_t^{l}(\mathcal{D}_S),\
	\end{gather*}
	follows. The matching upper bound follows by sending optimal linear index codes corresponding to all the sub-digraphs. 
	Thus, we have proved the statement $(i)$ of the theorem.
	
	Taking the limit as $t\to\infty$ in the definition of $\beta^l(\mathcal{D},\mathcal{P})$, we obtain (ii) in the statement of the theorem.
\end{proof}

We now illustrate the theorem using an example.
\begin{exmp}[Example \ref{exmp40} continued]
	Consider the three-sender unicast ICP given in Example \ref{exmp40}. The interaction digraph is shown in Figure \ref{fig3ICP}. All the interactions except the interactions $\mathcal{D}_{\{1,2,3\}} \rightarrow \mathcal{D}_{\{1,2\}}$, $\mathcal{D}_{\{1,2\}} \rightarrow \mathcal{D}_{\{2,3\}}$, $\mathcal{D}_{3} \rightarrow \mathcal{D}_{\{2,3\}}$, $\mathcal{D}_{2} \rightarrow \mathcal{D}_{\{2,3\}}$, and $\mathcal{D}_{1} \rightarrow \mathcal{D}_{\{2,3\}}$ are fully-participated. Observe that the interactions between ($\mathcal{D}_1$ and $\mathcal{D}_2$), and  ($\mathcal{D}_2$ and $\mathcal{D}_3$) are cyclic, which satisfy the condition $(iii)$ of Theorem \ref{thm3ICP2}. The interactions between any one among $\mathcal{D}_1$, $\mathcal{D}_2$, and $\mathcal{D}_3$, and other sub-digraphs are unidirectional, and are in the direction as given by the theorem, which satisfy condition  $(i)$. The interactions  $\mathcal{D}_{\{1,2,3\}} \rightarrow \mathcal{D}_{\{1,2\}}$ and  $\mathcal{D}_{\{1,2\}} \rightarrow \mathcal{D}_{\{2,3\}}$   do not form a cycle, thereby satisfying condition $(ii)$. Hence, the optimal linear broadcast rate is the sum of those of the individual sub-digraphs. In this example, for any $t$, we have $\beta^{l}_{t}(\mathcal{D}_{S})=1$ for all $S$ such that $\mathcal{P}_{S} \neq \Phi$. Hence, according to the theorem $\beta^{l}_{t}(\mathcal{D},\mathcal{P})=6$. The optimal transmissions are given as follows:
	\[  {\bf{x}}_1 ~\mbox{sent by $S_1$}, ~ {\bf{x}}_2  ~\mbox{sent by $S_2$}, ~{\bf{x}}_3 ~ \mbox{sent by $S_3$}, \]
	\[ {\bf{x}}_7  ~\mbox{sent by any one of $S_1$, $S_2$ and $S_3$}, \]
	\[ {\bf{x}}_5 \oplus  {\bf{x}}_6  ~ \mbox{sent by any one of $S_2$ and $S_3$,}
	\]
	\[
	{\bf{x}}_4 \oplus  {\bf{x}}_8  ~ \mbox{sent by any one of $S_1$ and $S_2$.}
	\]
	It can be easily verified that all the receivers can decode their demands using the broadcast transmissions  and their side-information.
	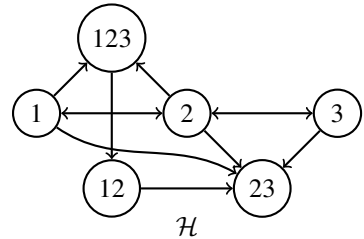
\begin{figure}
		\centering
		\begin{tikzpicture}
		[place/.style={circle,draw=black!100,thick}]
		\node at (-2,0) [place] (1c) {1};
		\node at (0,0) [place] (2c) {2};
		\node at (2,0) [place] (3c) {3};
		\node at (-1,-1) [place] (12c) {12};
		\node at (1,-1) [place] (23c) {23};
		\draw (0,-1.5) node {$\mathcal{H}$};
		\node at (-1,1) [place] (123c) {123};
		\draw [thick,<->] (1c.east) -- (2c.west);
		\draw [thick,<->] (2c.east) -- (3c.west);
		\draw [thick,->] (1c) to [out=45,in=-135] (123c);
		\draw [thick,->] (2c) to [out=135,in=-45] (123c);
		\draw [thick,->] (2c) to [out=-45,in=135] (23c);
		\draw [thick,->] (12c) to [out=0,in=180] (23c);
		\draw [thick,->] (123c) to [out=-90,in=90] (12c);
		\draw [thick,->] (1c) to [out=-35,in=150] (23c);
		\draw [thick,->] (3c) to [out=-135,in=45] (23c);
		\end{tikzpicture}
		\caption{Interaction digraph of Example \ref{exmp40} to illustrate Theorem \ref{thm3ICP2}.}
		\label{fig3ICP}
	\end{figure}
	\label{exmp4}
\end{exmp}

\subsection{Application of the proof techniques and results of the two-sender problem to the multi-sender ICP}
We consider an example of a three-sender unicast ICP, and show that the proof techniques used to solve the two-sender problem can be applied to obtain the optimal linear broadcast rate of the multi-sender problem for any number of senders. 

We now state and prove a required lemma to be used in the following example.
\begin{lem}
	Consider $n$ single-sender unicast problems, where the $i$th single-sender unicast problem has the side-information digraph $\mathcal{D}^{(i)}$, $i \in [n]$. If there are directed edges from every vertex of every digraph $\mathcal{D}^{(i)}$ to every vertex of every other digraph $\mathcal{D}^{(j)}$, $i \neq j$, $i,j \in [n]$, forming the digraph $\mathcal{D}$, then, $\beta^{l}_{t}(\mathcal{D}) = max \{\beta^{l}_{t}(\mathcal{D}^{(1)}),\cdots,\beta^{l}_{t}(\mathcal{D}^{(n)})\}$.
	\label{lemcliquedigraph}
\end{lem}
\begin{proof}
	We provide a lower bound on $\beta^{l}_{t}(\mathcal{D})$ and then provide a matching upper bound.
	
	Consider $\mathcal{D}$ and delete all the vertices and the incident edges belonging to all the digraphs $\mathcal{D}^{(j)}$ in $\mathcal{D}$ except $\mathcal{D}^{(i)}$, $i \neq j$, $i,j \in [n]$, for any $i$. As $\mathcal{D}^{(i)}$ is an induced sub-digraph of $\mathcal{D}$, we see that $\beta^{l}_{t}(\mathcal{D}) \geq \beta^{l}_{t}(\mathcal{D}^{(i)})$, as the optimal linear code for $\mathcal{D}$ also serves as an index code for $\mathcal{D}^{(i)}$ by setting all the messages belonging to all other $\mathcal{D}^{(j)}$, $j \neq i, j \in [n]$, to zero. As this is true for any $i$, we have $\beta^{l}_{t}(\mathcal{D}) \geq max \{\beta^{l}_{t}(\mathcal{D}^{(1)}),\cdots,\beta^{l}_{t}(\mathcal{D}^{(n)})\}$.
	
	Consider the linear index code given by $\mathcal{C}=\mathcal{C}_1 \oplus \mathcal{C}_2 \oplus  \cdots  \oplus \mathcal{C}_n$, where $\mathcal{C}_i$ is a linear index code for $\mathcal{D}^{(i)}$ with linear broadcast rate $\beta^{l}_{t}(\mathcal{D}^{(i)})$. It can be eaily verified that every receiver can decode its demands as every receiver in $\mathcal{D}^{(i)}$, $\forall i \in [n]$, knows every message in $\mathcal{D}^{(j)}$, $\forall j \neq i, j \in [n]$. The linear broadcast rate of this code is same as that given by the lower bound. Hence we have the result.	
\end{proof} 

\begin{exmp} 
	Consider any three-sender unicast ICP with the interaction digraph as given in Figure \ref{fig3ICP1}.  The interactions $\mathcal{D}_{1} \rightarrow \mathcal{D}_{\{1,2\}}$, $\mathcal{D}_{1} \rightarrow \mathcal{D}_{\{1,2,3\}}$, $\mathcal{D}_{\{1,2\}} \rightarrow \mathcal{D}_{\{1,2,3\}}$,$\mathcal{D}_{\{1,2\}} \rightarrow \mathcal{D}_{1}$, $\mathcal{D}_{\{1,2,3\}} \rightarrow \mathcal{D}_{1}$, $\mathcal{D}_{\{1,2,3\}} \rightarrow
	\mathcal{D}_{\{1,2\}}$, $\mathcal{D}_{\{2,3\}} \rightarrow \mathcal{D}_{3}$, and  $\mathcal{D}_{3} \rightarrow \mathcal{D}_{\{2,3\}}$ are fully-participated. Other interactions can be of any type.
	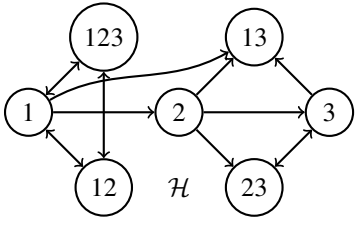
\begin{figure}
		\centering
		\begin{tikzpicture}
		[place/.style={circle,draw=black!100,thick}]
		\node at (-2,0) [place] (1c) {1};
		\node at (0,0) [place] (2c) {2};
		\node at (2,0) [place] (3c) {3};
		\node at (-1,-1) [place] (12c) {12};
		\node at (1,-1) [place] (23c) {23};
		\node at (1,1) [place] (13c) {13};
		\node at (-1,1) [place] (123c) {123};
		\draw (0,-1) node {$\mathcal{H}$};
		\draw [thick,->] (1c.east) -- (2c.west);
		\draw [thick,->] (2c.east) -- (3c.west);
		\draw [thick,->] (2c) to [out=-45,in=135] (23c);
		\draw [thick,->] (2c) to [out=45,in=-135] (13c);
		\draw [thick,<->] (3c) to [out=-135,in=45] (23c);
		\draw [thick,<->] (1c) to [out=45,in=-135] (123c);
		\draw [thick,<->] (1c) to [out=-45,in=135] (12c);
		\draw [thick,->] (3c) to [out=135,in=-45] (13c);
		\draw [thick,<->] (123c) to [out=-90,in=90] (12c);
		\draw [thick,->] (1c) to [out=30,in=-150] (13c);	
		\end{tikzpicture}
		\caption{Interaction digraph of Example \ref{exmp6}.}
		\label{fig3ICP1}
	\end{figure}
	
	We first obtain a lower bound on the optimal broadcast rate
	with $t$-bit messages, for any finite $t$. Considering the problem as a single-sender problem, we observe that there are unidirectional edges between the sub-digraph of the interaction digraph induced by the vertices $1$, ${\{1,2\}}$, and  ${\{1,2,3\}}$, and the sub-digraph of the interaction digraph induced by other vertices. Similarly, the edges $(2,{\{13\}})$, $(3,{\{13\}})$, $(2,{\{23\}})$, and $(2,3)$ do not lie on any cycle. Hence, we can remove the edges $(2,{\{13\}})$, $(3,{\{13\}})$, $(2,{\{23\}})$, and $(2,3)$ according to Lemma \ref{lemscc}, and still obtain the same optimal linear broadcast rate as that of the original problem  by analyzing the remaining problem. From Lemma \ref{lemcliquedigraph}, we have, 
	\begin{gather*}
	\beta^{l}_{t}(\mathcal{D}) = max \{\beta^{l}_{t}(\mathcal{D}_{1}),\beta^{l}_{t}(\mathcal{D}_{\{1,2\}}),\beta^{l}_{t}(\mathcal{D}_{\{1,2,3\}})\} \\ + max \{\beta^{l}_{t}(\mathcal{D}_{3}),\beta^{l}_{t}(\mathcal{D}_{\{2,3\}})\} + \beta^{l}_{t}(\mathcal{D}_{2}) + \beta^{l}_{t}(\mathcal{D}_{\{1,3\}}).
	\end{gather*}
	Note that this forms a lower bound for the optimal linear broadcast rate of the three-sender problem according to Lemma \ref{lowerbnd}.
	
	We obtain a matching upper bound, to obtain the optimal linear broadcast rate. Let $\mathcal{C}_{S}$ be the optimal linear code with linear broadcast rate $\beta^{l}_{t}(\mathcal{D}_{S})$, for the problem with side-information digraph $\mathcal{D}_{S}$. The transmissions are as follows. 
	\[  \mathcal{C}_1 \oplus \mathcal{C}_{\{1,2\}} \oplus \mathcal{C}_{\{1,2,3\}}  ~\mbox{sent by $S_1$}, \]
	\[  \mathcal{C}_3 \oplus \mathcal{C}_{\{2,3\}} ~ \mbox{sent by $S_3$}, ~ \mathcal{C}_2 ~\mbox{sent by $S_2$}, ~~\mbox{and} \]
	\[
	\mathcal{C}_{\{1,3\}} ~ \mbox{sent by any one of $S_1$ and $S_3$.}
	\]
	It can be easily verified that the linear broadcast rate of this code is same as that given by the lower bound. Also, it can be easily seen that the transmissions allow all the receivers to decode their demands. Hence, we obtain the optimal linear broadcast rate with $t$-bit messages to be,
	\begin{gather*} \beta^{l}_{t}(\mathcal{D},\mathcal{P}) = max \{\beta^{l}_{t}(\mathcal{D}_{1}),\beta^{l}_{t}(\mathcal{D}_{\{1,2\}}),\beta^{l}_{t}(\mathcal{D}_{\{1,2,3\}})\} \\ + max \{\beta^{l}_{t}(\mathcal{D}_{3}),\beta^{l}_{t}(\mathcal{D}_{\{2,3\}})\} + \beta^{l}_{t}(\mathcal{D}_{2}) + \beta^{l}_{t}(\mathcal{D}_{\{1,3\}}).
	\end{gather*}
	
	\begin{figure}
		\centering
		\begin{tikzpicture}
		[place/.style={circle,draw=black!100,thick}]
		\node at (-5,0) [place] (1c) {1};
		\node at (-2,0) [place] (2c) {2};
		\node at (0,0) [place] (3c) {3};
		\node at (-3,-1) [place] (4c) {4};
		\node at (-3,-2.2) [place] (8c) {8}; 	 	
		\node at (-1,-1) [place] (5c) {5};
		\node at (-1,1) [place] (9c) {9};	
		\node at (.5,1) [place] (10c) {10};
		\node at (-1,-2.2) [place] (6c) {6};	 	
		\draw (-2,-3) node {$\mathcal{D}$};
		\node at (-3,2) [place] (7c) {7};
		\draw [thick,->] (1c.east) -- (2c.west);
		\draw [thick,->] (2c.east) -- (3c.west);
		\draw [thick,->] (5c) to [out=-75,in=75] (6c); 
		\draw [thick,->] (6c) to [out=105,in=-105] (5c); 
		\draw [thick,->] (4c.south) -- (8c.north);
		\draw [thick,->] (7c) to [out=-180,in=90] (1c);
		\draw [thick,->] (1c) to [out=70,in=-160] (7c);
		\draw [thick,->] (1c) to [out=-35,in=125] (4c);
		\draw [thick,->] (4c) to [out=145,in=-55] (1c);
		\draw [thick,->] (2c) to [out=45,in=-135] (9c);
		\draw [thick,->] (3c) to [out=135,in=-45] (9c);
		\draw [thick,->] (1c) to [out=-75,in=155] (8c);
		\draw [thick,->] (8c) to [out=-180,in=-90] (1c);
		\draw [thick,->] (7c) to [out=-50,in=55,distance=.7cm] (8c);
		\draw [thick,->] (8c) to [out=115,in=-115] (7c);
		\draw [thick,->] (2c) to [out=-45,in=135] (5c);
		\draw [thick,->] (7c) to [out=-100,in=100] (4c);
		\draw [thick,->] (4c) to [out=75,in=-75] (7c);
		\draw [thick,->] (1c) to [out=15,in=-160] (9c);
		\draw [thick,->] (10c) to [out=160,in=20] (9c);
		\draw [thick,->] (9c) to [out=-20,in=-160] (10c);
		\draw [thick,->] (3c) to [out=-105,in=45] (6c);
		\draw [thick,->] (6c) to [out=25,in=-85] (3c);	
		\draw [thick,->] (3c) to [out=-120,in=30] (5c);
		\draw [thick,->] (5c) to [out=60,in=-150] (3c);
		\node at (-4,-4) [place] (1'c) {1};
		\draw (-4,-4.7) node {$\mathcal{D}_1$};
		\node at (-3,-4)  [place] (2'c) {2};
		\draw (-3,-4.7) node {$\mathcal{D}_2$};
		\node at (-2,-4)  [place] (3'c) {3};	
		\draw (-2,-4.7) node {$\mathcal{D}_3$};
		\node at (0,-4)  [place] (139'c) {9};	
		\node at (1.2,-4)  [place] (1310'c) {10};
		\draw (.6,-4.7) node {$\mathcal{D}_{\{1,3\}}$};
		\draw [thick,->] (139'c) to [out=20,in=160] (1310'c);
		\draw [thick,->] (1310'c) to [out=-160,in=-20] (139'c);
		\node at (-1,-4)  [place] (123'c) {7};	
		\draw (-1,-4.7) node {$\mathcal{D}_{\{1,2,3\}}$};
		\node at (-4,-5.5) [place] (124'c) {4};
		\node at (-2.5,-5.5) [place] (128'c) {8};
		\draw [thick,->] (124'c) to [out=0,in=180] (128'c);	
		\draw (-3.3,-6) node {$\mathcal{D}_{\{1,2\}}$};
		\node at (-1.5,-5.5) [place] (235'c) {5};
		\node at (0,-5.5) [place] (236'c) {6};	 
		\draw [thick,->] (235'c) to [out=20,in=160] (236'c);	
		\draw [thick,->] (236'c) to [out=-160,in=-20] (235'c);
		\draw (-.7,-6) node {$\mathcal{D}_{\{2,3\}}$}; 	 		
		\end{tikzpicture}
		\caption{Side-information digraph and sub-digraphs of Example \ref{exmp6}}
		\label{fig3ICP6}
	\end{figure}
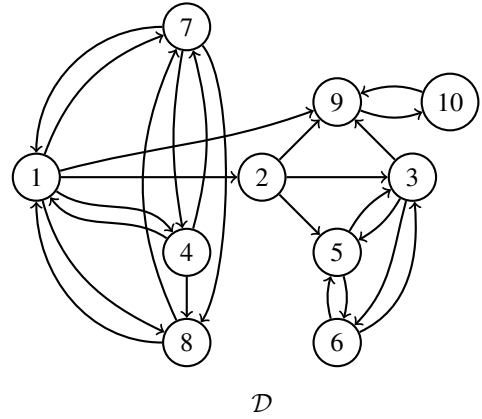
	Now, we consider a three-sender unicast ICP with the interaction digraph as given in the example. There are $m=10$ messages and the message sets at the senders are given by   $\mathcal{M}_1=\{{\bf{x}}_1,{\bf{x}}_4,{\bf{x}}_7,{\bf{x}}_8,{\bf{x}}_9,{\bf{x}}_{10}\}$, $\mathcal{M}_2=\{{\bf{x}}_2,{\bf{x}}_4,{\bf{x}}_5,{\bf{x}}_6,{\bf{x}}_7,{\bf{x}}_8\}$, and $\mathcal{M}_3=\{{\bf{x}}_3,{\bf{x}}_5,{\bf{x}}_6,{\bf{x}}_7,{\bf{x}}_9,{\bf{x}}_{10}\}$. From the definition of $\mathcal{P}_{S}$ for any non-empty $S \subseteq [s]$, we have the following:
	\begin{gather*}
	\mathcal{P}_1 = \{{\bf{x}}_1\}, \mathcal{P}_2 = \{{\bf{x}}_2\}, \mathcal{P}_3 = \{{\bf{x}}_3\}, \mathcal{P}_{\{1,2\}} = \{{\bf{x}}_4,{\bf{x}}_8\}, \\ \mathcal{P}_{\{2,3\}}=\{{\bf{x}}_5,{\bf{x}}_6\}, \mathcal{P}_{\{1,3\}}=\{{\bf{x}}_9,{\bf{x}}_{10}\},\mathcal{P}_{\{1,2,3\}}=\{{\bf{x}}_7\},\\
	\mathcal{P}=(\mathcal{P}_1,\mathcal{P}_2,\mathcal{P}_3,\mathcal{P}_{\{1,2\}},\mathcal{P}_{\{1,3\}},\mathcal{P}_{\{2,3\}},\mathcal{P}_{\{1,2,3\}}).
	\end{gather*}
	The  side-information of the receivers are given as follows:
	\begin{gather*}
	\mathcal{K}_1 = \{{\bf{x}}_2,{\bf{x}}_4,{\bf{x}}_7,{\bf{x}}_8,{\bf{x}}_9\}, \mathcal{K}_2 = \{{\bf{x}}_3,{\bf{x}}_5,{\bf{x}}_9\},\\ \mathcal{K}_3 = \{{\bf{x}}_5,{\bf{x}}_6,{\bf{x}}_9\}, \mathcal{K}_{4} = \{{\bf{x}}_1,{\bf{x}}_7,{\bf{x}}_8\},  \mathcal{K}_{5}=\{{\bf{x}}_3,{\bf{x}}_6\},\\ \mathcal{K}_{6}=\{{\bf{x}}_3,{\bf{x}}_5\},\mathcal{K}_{7}=\{{\bf{x}}_1,{\bf{x}}_4,{\bf{x}}_8\},\\\mathcal{K}_{8}=\{{\bf{x}}_1,{\bf{x}}_7\},\mathcal{K}_{9}=\{{\bf{x}}_{10}\},\mathcal{K}_{10}=\{{\bf{x}}_{9}\}.
	\end{gather*}
	The side-information digraph is shown in Figure \ref{fig3ICP6}. Observe that the interactions $\mathcal{D}_{1} \rightarrow \mathcal{D}_{\{1,2\}}$, $\mathcal{D}_{1} \rightarrow \mathcal{D}_{\{1,2,3\}}$, $\mathcal{D}_{\{1,2\}} \rightarrow \mathcal{D}_{\{1,2,3\}}$,$\mathcal{D}_{\{1,2\}} \rightarrow \mathcal{D}_{1}$, $\mathcal{D}_{\{1,2,3\}} \rightarrow \mathcal{D}_{1}$, $\mathcal{D}_{\{1,2,3\}} \rightarrow \mathcal{D}_{\{1,2\}}$, $\mathcal{D}_{\{2,3\}} \rightarrow \mathcal{D}_{3}$, and  $\mathcal{D}_{3} \rightarrow \mathcal{D}_{\{2,3\}}$ are fully-participated. Note that $\beta^{l}_{t}(\mathcal{D}_{\{1,2\}})=2$ and optimal linear broadcast rate of all other sub-digraphs is equal to $1$. Hence,
	\[
	\beta^{l}_{t}(\mathcal{D},\mathcal{P}) = max \{1,2,1\} + max \{1,1\} + 1 + 1 = 5.
	\]
	The optimal codes for the sub-digraphs are as follows : 
	\begin{gather*}
	\mathcal{C}_1 = {\bf{x}}_1, \mathcal{C}_2 = {\bf{x}}_2, \mathcal{C}_3 = {\bf{x}}_3, \mathcal{C}_{\{1,2\}} = ({\bf{x}}_4,{\bf{x}}_8), \\ 
	\mathcal{C}_{\{1,3\}} = {\bf{x}}_9 \oplus {\bf{x}}_{10}, \mathcal{C}_{\{2,3\}} = {\bf{x}}_5 \oplus {\bf{x}}_{6},  \mathcal{C}_{\{1,2,3\}} = {\bf{x}}_7.
	\end{gather*}
	
	\label{exmp6}
\end{exmp}

When it is difficult to use the proof techniques of the two-sender problem, partitioning the interaction-digraph and using the optimal linear broadcast rate of the partitions will provide sub-optimal solutions for the multi-sender problem.

\section{Conclusion and Future Work}   
\par This paper establishes the optimal linear broadcast rates and optimal linear code constructions for all the cases of the TUICP with fully-participated interactions. The results for the TUICP are given in terms of those of the three sub-problems which are single-sender unicast ICPs. These results are significant as they provide non-trivial lower bounds for the optimal linear broadcast rates of two-sender problems with  partially-participated interactions.

\par In \cite{CBSR}, optimal scalar linear codes were obtained for Cases II-C and II-D for a special class of  partially-participated interactions.  Finding the optimal linear broadcast rates and corresponding codes for two-sender problems belonging to Cases II-B, II-C, II-D, and II-E with partially-participated interactions is an interesting open problem. 

Further, extension of the results to multi-sender problems with more than two senders, to obtain the optimal linear broadcast rates and code constructions is interesting.

\section*{Acknowledgment}
This work was supported partly by the Science and Engineering Research Board (SERB) of Department of Science and Technology (DST), Government of India, through J.C. Bose National Fellowship to B. S. Rajan and VAJRA Fellowship to V. Aggarwal.


\begin{thebibliography}{10}
	
\bibitem{BK}
Y.~Birk and T.~Kol, ``Coding on demand by an informed source ({ISCOD}) for
efficient broadcast of different supplemental data to caching clients,''
\emph{IEEE Transactions on Information Theory}, vol.~52, no.~6, pp.
2825--2830, 2006.

\bibitem{KAC}
K.~Shanmugam, N.~Golrezaei, A.~G. Dimakis, A.~F. Molisch, and G.~Caire,
``Femtocaching: Wireless content delivery through distributed caching
helpers,'' \emph{IEEE Transactions on Information Theory}, vol.~59, no.~12,
pp. 8402--8413, 2013.

\bibitem{luo2016coded}
T.~Luo, V.~Aggarwal, and B.~Peleato, ``Coded caching with distributed
storage,'' \emph{arXiv preprint arXiv:1611.06591}, 2016.

\bibitem{xiang2016joint}
Y.~Xiang, T.~Lan, V.~Aggarwal, and Y.-F.~R. Chen, ``Joint latency and cost
optimization for erasure-coded data center storage,'' \emph{IEEE/ACM
	Transactions on Networking (TON)}, vol.~24, no.~4, pp. 2443--2457, 2016.

\bibitem{SUOH}
L.~Ong, C.~K. Ho, and F.~Lim, ``The single-uniprior index-coding problem: The
single-sender case and the multi-sender extension,'' \emph{IEEE Transactions
	on Information Theory}, vol.~62, no.~6, pp. 3165--3182, 2016.

\bibitem{COJ}
C.~Thapa, L.~Ong, and S.~J. Johnson, ``Graph-theoretic approaches to two-sender
index coding,'' in \emph{IEEE Globecom Workshops}, 2016, pp. 1--6.

\bibitem{sadeghi2016distributed}
P.~Sadeghi, F.~Arbabjolfaei, and Y.-H. Kim, ``Distributed index coding,'' in
\emph{IEEE Information Theory Workshop (ITW)}, 2016, pp. 330--334.

\bibitem{YPFK}
Y.~Liu, P.~Sadeghi, F.~Arbabjolfaei, and Y.-H. Kim, ``On the capacity for
distributed index coding,'' in \emph{IEEE International Symposium on
	Information Theory (ISIT)}, 2017, pp. 3055--3059.

\bibitem{MOJ2}
M.~Li, L.~Ong, and S.~J. Johnson, ``Improved bounds for multi-sender index
coding,'' in \emph{IEEE International Symposium on Information Theory
	(ISIT)}, 2017, pp. 3060--3064.

\bibitem{MOJ}
------, ``Cooperative multi-sender index coding,'' \emph{arXiv preprint
	arXiv:1701.03877}, 2017.

\bibitem{CTLO}
C.~Thapa, L.~Ong, S.~J. Johnson, and M.~Li, ``Structural characteristics of
two-sender index coding,'' \emph{arXiv preprint arXiv:1711.08150}, 2017.

\bibitem{fekete}
M.~Fekete, ``Uber die verteilung der wurzeln bei gewissen algebraischen
gleichungen mit ganzzahligen koeffizienten,'' \emph{Mathematische
	Zeitschrift}, vol.~17, pp. 228--249, 1923.

\bibitem{DBW}
D.~B. West, \emph{Introduction to graph theory}.\hskip 1em plus 0.5em minus
0.4em\relax Prentice hall Upper Saddle River, 2001, vol.~2.

\bibitem{DEK}
D.~E. Knuth, \emph{The Art of Computer Programming}.\hskip 1em plus 0.5em minus
0.4em\relax Addition-Wesley, 1968, vol.~1.

\bibitem{tahmasbi2015critical}
M.~Tahmasbi, A.~Shahrasbi, and A.~Gohari, ``Critical graphs in index coding,''
\emph{IEEE Journal on Selected areas in Communications}, vol.~33, no.~2, pp.
225--235, 2015.

\bibitem{jafar}
S.~A. Jafar, ``Interference alignment- a new look at signal dimensions in a
communication network,'' \emph{Foundations and Trends in Communications and
	Information Theory}, vol.~7, pp. 1--136, 2011.

\bibitem{CBSR}
C.~Arunachala and B.~S. Rajan, ``Optimal scalar linear index codes for three
classes of two-sender unicast index coding problem,'' \emph{arXiv preprint
	arXiv:1804.03823}, 2018.

\end{thebibliography}
\end{document}